\documentclass[journal,10pt,twoside]{IEEEtran}
 
 \usepackage{tikz} \usetikzlibrary{positioning,shapes.geometric,arrows.meta,fit,backgrounds,calc}
 \definecolor{cIn}{RGB}{237,239,242}
 \definecolor{cRisk}{RGB}{249,241,219}
 \definecolor{cRiskB}{RGB}{171,131,28}
 \definecolor{cSub}{RGB}{214,227,243}
 \definecolor{cS}{RGB}{220,237,223}
 \definecolor{cP}{RGB}{243,216,192}
 \definecolor{cMerge}{RGB}{228,222,240}
 \definecolor{cOut}{RGB}{208,210,215}
 \definecolor{cEng}{RGB}{250,250,251}
 \definecolor{cTitle}{RGB}{55,57,64}
 \definecolor{cRef}{RGB}{255,255,255}
 
\usepackage{hyperref}
\hypersetup{colorlinks=true,	linkcolor=black,	citecolor=blue,	}
\usepackage[acronym,shortcuts,automake,nogroupskip,nopostdot]{glossaries}
\makeglossaries
\usepackage{amsmath,amssymb,amsthm}
\usepackage{bm}
\usepackage{graphicx}
\usepackage[compatibility=false]{caption}
\usepackage{subcaption}
\usepackage{svg}
\usepackage{cite}
\usepackage{algorithm}
\usepackage{algpseudocode}
\usepackage{pifont}
\newcommand{\xmark}{\ding{55}}

\algnewcommand\algorithmicinput{\textbf{Input:}}
\algnewcommand\algorithmicoutput{\textbf{Output:}}
\algnewcommand\Input{\item[\algorithmicinput]}
\algnewcommand\Output{\item[\algorithmicoutput]}
\algrenewcommand\algorithmiccomment[1]{\hfill$\triangleright$\,\textit{\small #1}}
\usepackage{booktabs}
\usepackage{multirow}
\usepackage{array}
\usepackage{xcolor}
\usepackage{enumitem}

\newtheorem{proposition}{Proposition}

\newtheorem{definition}{Definition}

\newcommand{\calS}{\mathcal{S}}
\newcommand{\calN}{\mathcal{N}}
\newcommand{\calF}{\mathcal{F}}
\newcommand{\calI}{\mathcal{I}}
\newcommand{\calU}{\mathcal{U}}
\newcommand{\calT}{\mathcal{T}}
\newcommand{\bfbeta}{\bm{\beta}}
\newcommand{\tbeta}{\tilde{\bm{\beta}}}
\newcommand{\hbeta}{\hat{\bm{\beta}}}
\newcommand{\bfgamma}{\bm{\gamma}}
\newcommand{\Risk}{\mathcal{R}}

\begin{document}
\setlength{\parskip}{0pt}
\setlength{\parindent}{15pt}
\newacronym{gsm}{GSM}{global system for mobile communication}
\newacronym{gsma}{GSMA}{global system for mobile communication association}
\newacronym{sdo}{SDO}{standardization development organization}
\newacronym{ietf}{IETF}{Internet engineering task force}
\newacronym{itu}{ITU}{international telecommunication union}
\newacronym{1g}{1G}{first generation}
\newacronym{2g}{2G}{second generation}
\newacronym{3g}{3G}{third generation}
\newacronym{4g}{4G}{fourth generation}
\newacronym{5g}{5G}{fifth generation}
\newacronym{5gs}{5GS}{5G system}
\newacronym{5gc}{5GC}{5G core}
\newacronym{sdn}{SDN}{software-defined network}
\newacronym{nfv}{NFV}{network function virtualization}
\newacronym{xnap}{XnAP}{Xn application protocol}
\newacronym{3gpp}{3GPP}{3rd generation partnership project}
\newacronym{nr}{NR}{new radio}
\newacronym{ue}{UE}{user equipment}
\newacronym{nea}{NEA}{new radio encryption algorithm}
\newacronym{nia}{NIA}{new radio integrity algorithm}
\newacronym{mec}{MEC}{mobile edge computing}
\newacronym{sim}{SIM}{subscriber identification module}
\newacronym{supi}{SUPI}{subscription permanent identifier}
\newacronym{nas}{NAS}{non-access stratum}
\newacronym{as}{AS}{access stratum}
\newacronym{an}{AN}{access network}
\newacronym{ngmn}{NGMN}{next generation mobile networks alliance}
\newacronym{etsi}{ETSI}{european telecommunications standards institute}
\newacronym{scp}{SCP}{service communication proxy}
\newacronym{mar}{MAR}{multi-access rules}
\newacronym{far}{FAR}{forwarding action rule}
\newacronym{nhcc}{NhCC}{next hop chaining counter}
\newacronym{ie}{IE}{information element}
\newacronym{ipups}{IPUPS}{inter-PLMN \gls{up} security}
\newacronym{cups}{CUPS}{control and user plane separation}
\newacronym{smc}{SMC}{security mode command}
\newacronym{aka}{AKA}{authentication and key agreement}
\newacronym{peaa}{PEAA}{privacy-enhanced access authentication}
\newacronym{eap-aka'}{EAP-AKA'}{improved \gls{eap}-\gls{aka}}
\newacronym{suci}{SUCI}{subscription concealed identifier}
\newacronym{ausf}{AUSF}{authentication server function}
\newacronym{arpf}{ARPF}{authentication credential repository and processing function}
\newacronym{sidf}{SIDF}{subscription identifier de-concealing function}
\newacronym{gtp-u}{GTP-U}{GPRS tunnelling protocol - user}
\newacronym{gtp-c}{GTP-C}{GPRS tunnelling protocol - control}
\newacronym{seaf}{SEAF}{security anchor function}
\newacronym{kseaf}{KSEAF}{key security anchor function}
\newacronym{guti}{GUTI}{globally unique temporary identity}
\newacronym{pdcp}{PDCP}{packet data convergence protocol}
\newacronym{rrc}{RRC}{radio resource control}
\newacronym{pdu}{PDU}{protocol data unit}
\newacronym{pfcp}{PFCP}{packet forwarding control protocol}
\newacronym{sba}{SBA}{service based architecture}
\newacronym{sbi}{SBI}{service based interface}
\newacronym{plmn}{PLMN}{public land mobile network}
\newacronym{hplmn}{H-PLMN}{home \gls{plmn}}
\newacronym{vplmn}{V-PLMN}{visited \gls{plmn}}
\newacronym{splmn}{S-PLMN}{serving \gls{plmn}}
\newacronym{sepp}{SEPP}{security edge protection proxy}
\newacronym{ipx}{IPX}{IP exchange service}
\newacronym{oauth}{OAuth}{open authorization}
\newacronym{nf}{NF}{network function}
\newacronym{mno}{MNO}{mobile network operator}
\newacronym{upf}{UPF}{user plane function}
\newacronym{up}{UP}{user plane}
\newacronym{5g-ia7}{5G-IA7}{5G-integrity algorithm 7}
\newacronym{udp}{UDP}{user datagram protocol}
\newacronym{tcp}{TCP}{transport control protocol}
\newacronym{nrf}{NRF}{network repository function}
\newacronym{ipsec}{IPsec}{Internet protocol security}
\newacronym{udm}{UDM}{unified data management}
\newacronym{l1}{L1}{layer 1}
\newacronym{l2}{L2}{layer 2}
\newacronym{l3}{L3}{layer 3}
\newacronym{l4}{L4}{layer 4}
\newacronym{ran}{RAN}{radio access network}
\newacronym{cn}{CN}{core network}
\newacronym{oran}{O-RAN}{open radio access network}
\newacronym{ranap}{RANAP}{radio access network application part}
\newacronym{ecpri}{eCPRI}{evolved common public radio interface}
\newacronym{vxlan}{VxLAN}{virtual extensible local area network}
\newacronym{ssh}{SSH}{secure shell protocol}
\newacronym{pq}{PQ}{post-quantum}
\newacronym{iot}{IoT}{Internet of things}
\newacronym{mitm}{MiTM}{man-in-the-middle}
\newacronym{pls}{PLS}{physical layer security}
\newacronym{phy}{PHY}{physical layer}
\newacronym{d2d}{D2D}{direct-to-device}
\newacronym{v2x}{v2x}{vehicle-to-everything}
\newacronym{m2m}{M2M}{machine to machine}
\newacronym{macsec}{MACsec}{media access control security}
\newacronym{siem}{SIEM}{security information and event management}
\newacronym{api}{API} {application programming interface}
\newacronym{capif}{CAPIF} { common \gls{api} framework}
\newacronym{dpi}{DPI} {deep packet inspection}
\newacronym{prins}{PRINS}{protocol for N32 interconnect security}
\newacronym{stride}{STRIDE}{spoofing, tampering, repudiation, information disclosure, denial of service, and elevation of privilege}
\newacronym{dos}{DoS}{denial of service}
\newacronym{cu}{CU}{central unit}
\newacronym{du}{DU}{distributed unit}
\newacronym{m-plane}{M-Plane}{management plane}
\newacronym{ecies}{ECIES}{elliptic curve integrated encryption scheme}
\newacronym{es256}{ES256}{\gls{sha}256 with elliptic curve digital signature algorithm (ECDSA)}
\newacronym{amf}{AMF}{access and mobility function}
\newacronym{smf}{SMF}{session management function}
\newacronym{pcf}{PCF}{policy control function}
\newacronym{nef}{NEF}{network exposure function}
\newacronym{af}{AF}{application function}
\newacronym{nssf}{NSSF}{network slicing selection function}
\newacronym{dn}{DN}{data network}
\newacronym{dnn}{DNN}{data network name}
\newacronym{aaa}{AAA}{authentication, authorization, and accounting}
\newacronym{tls}{TLS}{transport layer security}
\newacronym{dtls}{DTLS}{datagram transport layer security}
\newacronym{gnb}{gNB}{gNodeB}
\newacronym{bbu}{BBU}{baseband unit}
\newacronym{umts}{UMTS}{universal mobile telecommunications system}
\newacronym{utran}{UTRAN}{\gls{umts} terrestrial radio access network}
\newacronym{teid}{TEID}{tunnel endpoint identifier}
\newacronym{imsi}{IMSI}{international mobile subscriber identity}
\newacronym{mme}{MME}{mobility management entity}
\newacronym{ng}{NG}{next generation}
\newacronym{enb}{eNB}{enhanced NodeB}
\newacronym{sn}{SN}{serving network}
\newacronym{ike}{IKE}{Internet key exchange}
\newacronym{ikev2}{IKEv2}{Internet key exchange version 2}
\newacronym{esp}{ESP}{encapsulating security payload}
\newacronym{wesp}{WESP}{wrapped encapsulating security payload}
\newacronym{ng-ran}{NG-RAN}{\gls{ng}-\gls{ran}}
\newacronym{qos}{QoS}{quality of service}
\newacronym{qoe}{QoE}{quality of experience}
\newacronym{nds}{NDS}{network domain security}
\newacronym{ip}{IP}{Internet protocol}
\newacronym{json}{JSON}{JavaScript object notation}
\newacronym{jose}{JOSE}{JavaScript object signing and encryption}
\newacronym{seid}{SEID}{session endpoint identifier}
\newacronym{pei}{PEI}{permanent equipment identifier}
\newacronym{hres*}{HRES*}{hash response}
\newacronym{hxres*}{HXRES*}{hash expected response}
\newacronym{henb}{HeNB}{home \gls{enb}}
\newacronym{uicc}{UICC}{universal integrated circuit card}
\newacronym{rn}{RN}{radio network}
\newacronym{nssaa}{NSSAA}{network slice-specific authentication and authorization}
\newacronym{http/2}{HTTP/2}{hypertext transfer protocol version 2}
\newacronym{aes}{AES}{advanced encryption standard}
\newacronym{mac}{MAC}{message authentication code}
\newacronym{mediumac}{MAC}{medium access control}
\newacronym{pfs}{PFS}{perfect forward secrecy}
\newacronym{sha}{SHA}{secure hash algorithm}
\newacronym{hmac}{HMAC}{hash-based message authentication code}
\newacronym{ts}{TS}{technical specification}
\newacronym{ip-tfs}{IP-TFS}{\gls{ip} - \gls{tfs}}
\newacronym{sa}{SA}{simulated annealing}
\newacronym{gsa}{GSA}{group security association}
\newacronym{g-ikev2}{G-IKEv2}{group key management using \gls{ikev2}}
\newacronym{tfc}{TFC}{traffic flow confidentiality}
\newacronym{tfs}{TFS}{traffic flow security}
\newacronym{pki}{PKI}{public key infrastructure}
\newacronym{ota}{OTA}{over the air}
\newacronym{os}{OS}{operating system}
\newacronym{res}{RES}{response verification}
\newacronym{res*}{RES*}{response from RES}
\newacronym{urllc}{URLLC}{ultra reliable and low latency communications}
\newacronym{embb}{eMBB}{enhanced mobile broadband }
\newacronym{gtp}{GTP}{general packet radio services tunneling protocol}
\newacronym{ap}{AP}{application protocol}
\newacronym{cp}{CP}{control plane}
\newacronym{rlc}{RLC}{radio link control}
\newacronym{eap}{EAP}{extensible authentication protocol}
\newacronym{rb}{RB}{radio bearer}
\newacronym{apn}{APN}{access point name}
\newacronym{uu}{UU}{unique user}
\newacronym{cia}{CIA}{confidentiality, integrity, and availability}
\newacronym{nist}{NIST}{National Institute of Standards and Technology}
\newacronym{http}{HTTP}{hypertext transfer protocol}
\newacronym{sm}{SM}{session management}
\newacronym{jws}{JWS}{JSON web signature}
\newacronym{urr}{URR}{usage reporting rule}
\newacronym{bar}{BAR}{buffering action rule}
\newacronym{pdr}{PDR}{packet detection rule}
\newacronym{iab}{IAB}{integrated access and backhaul}
\newacronym{enisa}{ENISA}{european union agency for cybersecurity}
\newacronym{nssid}{NSSID}{network slice subnet instance identifier}
\newacronym{sctp}{SCTP}{stream control transmission protocol}
\newacronym{ns}{NS}{network slicing}
\newacronym{vnf}{VNF}{virtual network function}
\newacronym{imei}{IMEI}{international mobile equipment identity}
\newacronym{ml}{ML}{machine learning}
\newacronym{leo}{LEO}{Low Earth orbit}
\newacronym{cnn}{CNN}{convolutional neural network}
\newacronym{geo}{GEO}{geostationary orbit}
\newacronym{ntn}{NTN}{non-terrestrial network}
\newacronym{haps}{HAPS}{high altitude platform station}
\newacronym{ai}{AI}{artificial intelligence}
\newacronym{pep}{PEP}{performance enhancing proxies}
\newacronym{vpn}{VPN}{virtual private network}
\newacronym{ris}{RIS}{reconfigurable intelligent surfaces}
\newacronym{fl}{FL}{federated learning}
\newacronym{iov}{IoV}{Internet of vehicles}
\newacronym{sagin}{SAGIN}{space-air-ground integrated network}
\newacronym{istn}{6G-ISTN}{6G integrated satellite-terrestrial networks}
\newacronym{tn}{TN}{terrestrial networks}
\newacronym{hsn}{HSN}{hybrid satellite networks}
\newacronym{ibc}{IBC}{identity-based cryptography}
\newacronym{ncc}{NCC}{network control center}
\newacronym{mn}{MN}{mobile node}
\newacronym{vhn}{VHetNet}{vertical heterogeneous network}
\newacronym{pla}{PLA}{physical layer authentication}
\newacronym{ds}{DS}{Doppler frequency shift}
\newacronym{pkg}{PKG}{private key generator}
\newacronym{sgin}{SGIN}{space–ground integrated network}
\newacronym{pbs}{PBS}{perfect backward secrecy}
\newacronym{pfbs}{PFBS}{perfect forward/backward secrecy}
\newacronym{ecc}{ECC}{elliptic curve cryptography}
\newacronym{pkc}{PKC}{public key cryptography}
\newacronym{skc}{SKC}{secret key cryptography}
\newacronym{crt}{CRT}{chinese remainder theorem}
\newacronym{sdh}{SDH}{strong diffie-hellman}
\newacronym{ecdsa}{ECDSA}{elliptic curve digital signature algorithm}
\newacronym{ccsds}{CCSDS}{consultative committee for space data systems}
\newacronym{6g}{6G}{sixth-generation}
\newacronym{gnss}{GNSS}{global navigation satellite system}
\newacronym{sin}{SIN}{spatial information networks}
\newacronym{csi}{CSI}{channel state information}
\newacronym{los}{LOS}{line-of-sight}
\newacronym{iq}{I-Q}{in-phase and quadrature}
\newacronym{srp}{SRP}{signal received power}
\newacronym{snr}{SNR}{signal-to-noise-ratio}
\newacronym{isl}{ISL}{inter-satellite link}
\newacronym{satcom}{SamCom}{satellite communication}
\newacronym{sdls}{SDLS}{space data link security}
\newacronym{ah}{AH}{authentication header}
\newacronym{ids}{IDS}{intrusion detection system}
\newacronym{ips}{IPS}{intrusion prevention system}
\newacronym{vnr}{VNR}{virtual network request}
\newacronym{vne}{VNE}{virtual network embedding}
\newacronym{sfc}{SFC}{security function chain}
\newacronym{ssf}{SSF}{security service function}
\newacronym{sit}{SLSN}{multi-slice LEO satellite network}
\newacronym{mips}{MIPS}{million instructions per second}
\newacronym{sla}{SLA}{service level agreement}
\newacronym{sf}{SF}{security function}
\newacronym{milp}{MILP}{mixed-integer linear programming}
\newacronym{minlp}{MINLP}{mixed-integer non-linear programming}
\newacronym{fair}{G-Fair}{greedy association with fair bandwidth distribution}
\newacronym{dad}{B-DAD}{balanced association with delay-aware bandwidth distribution}
\newacronym{gs}{GS}{ground station}
\newacronym{ua}{UA}{user association}
\newacronym{dsr}{DSR}{demand satisfaction ratio}
\newacronym{cdf}{CDF}{cumulative distribution function}
\newacronym{sr}{Max-SR}{maximum sum of data rate}
\newacronym{vr}{VR}{virtual reality}
\newacronym{v2v}{V2V}{vehicle-to-vehicle}
\newacronym{raan}{RAAN}{right ascension of the ascending node}
\newacronym{kpi}{KPI}{key performance indicator}
\newacronym{sfa}{SFA}{security functions allocation}
\newacronym{e2e}{E2E}{end-to-end}
\newacronym{rl}{RL}{reinforcement learning}
\newacronym{fsol}{FSO}{free space optical}
\newacronym{admm}{ADMM}{alternating direction method of multipliers}
\newacronym{ru}{RU-MILP}{Risk-Unaware MILP}
\newacronym{ilp}{ILP}{integer linear programme}
\newacronym{fw}{FW}{firewall}
\newacronym{tm}{TM}{traffic monitor}

\title{Scalable Security and Migration-Aware SFC Provisioning in LEO Satellite Networks}

\author{
  Mohammed Mahyoub$^{1}$,
  Wael Jaafar$^{3}$,
  Sami Muhaidat$^{1,2}$,
  and Halim Yanikomeroglu$^{1}$
  \\
  $^{1}$Department of Systems and Computer Engineering, Carleton University,
        Ottawa, Ontario, Canada
  \\
  $^{2}$Department of Computer Science, Khalifa University, Abu Dhabi, UAE
  \\
  $^{3}$Department of Software and IT Engineering,
        \'{E}cole de Technologie Sup\'{e}rieure, Montreal, Quebec, Canada
}
\maketitle

\begin{abstract}
\gls{leo} satellite constellations are emerging as a backbone for global 6G connectivity, where independent tenant slices share orbital infrastructure, each requiring an ordered chain of security \glspl{vnf}. Because onboard computation and networking are scarce, slices cannot be given dedicated \glspl{vnf}. They must share instances on the same satellites, enlarging the attack surface and exposing tenants to cross-slice side-channel risk. This exposure shifts continually as visibility, orbital motion, and the inter-satellite topology change in time (epochs), making \gls{vnf} migration a structural necessity that couples resource efficiency, service continuity, and security isolation into a single problem. We formulate this security- and migration-aware \gls{sfc} placement as a multi-slice \gls{milp} whose core is a co-location risk model, grounded in ISO/NIST principles and supported by analytic bounds, in which we separate avoidable migrations from those forced by orbital motion. Because the joint program scales quadratically with the cross-slice co-location terms, we develop an \gls{admm}-inspired penalized per-slice best-response decomposition that recasts the coupling as a linear per-slice penalty, yielding independent subproblems through sequential (S-\gls{admm}) and parallel, collision-repaired (P-\gls{admm}) schedules. Simulations over a Walker-Delta satellite constellation show that the proposed framework eliminates co-location risk, reduces \gls{sfc} migrations, and sustains full delay compliance, while remaining feasible within the per-epoch budget for slice counts where the monolithic security-aware \gls{milp} is intractable. 
\end{abstract}

\begin{IEEEkeywords}
\gls{leo} satellites, \gls{sfc} placement, \gls{admm} decomposition, security-aware optimization, \gls{vnf} migration, 6G.
\end{IEEEkeywords}

\section{Introduction}
\label{sec:introduction}

\IEEEPARstart{L}{ow}  Earth orbit \gls{leo} satellite megaconstellations emerge as critical infrastructure for global 6G connectivity \cite{Abdelsadek2023}.
\gls{leo} megaconstellations such as SpaceX Starlink, Amazon Kuiper, and OneWeb place thousands of satellites at altitudes of 340-1200~km, delivering broadband coverage with low latency \cite{Pachler2024}.
The integration of \gls{leo} backhaul and access capacity into 5G and next-generation 6G radio access networks (RANs) has motivated a growing interest in on-orbit edge computing, offloading latency-sensitive and security-demanding processing from terrestrial data centres to satellite computing nodes \cite{Qiufen2024}.

In multi-slice deployments where independent organizations share the same orbital infrastructure, security processing for each slice is naturally realized as an ordered chain of software-defined virtual network functions  (\glspl{vnf}) \cite{etsi_nfv}, forming a \acrfull{sfc} that may include \gls{fw}, \gls{ids}, \gls{tm}, \gls{siem}, and encryption functions \cite{Mahyoub2026}.
The orchestration problem of deciding which satellite hosts which \gls{vnf} instance for which slice's user is the main \gls{sfc} placement problem studied in this paper.

This placement problem is significantly more challenging in orbital networks than in terrestrial environments due to the highly dynamic nature of \gls{leo} constellations. On one hand, unlike ground networks where topology changes are relatively infrequent, the connectivity graph in \gls{leo} systems evolves continuously as satellites move along their orbital trajectories \cite{mahyoub2026_visibility}. For example, satellites operating at an altitude of approximately $550$~km complete an orbit around the Earth in roughly $90$~minutes and remain visible to a given ground user for only $5$--$10$~minutes. Consequently, at each orchestration epoch, the network topology may differ substantially. Specifically, \gls{isl} links are reconfigured, end-to-end propagation delays and shortest-path routes vary, and the set of satellites capable of serving a particular user changes continuously \cite{Laniewski2025}. This persistent topology evolution tightly couples \gls{sfc} placement decisions with the temporal dynamics of the constellation, making service-chain orchestration more complex than in static or slowly varying terrestrial networks.
Consequently, epoch-by-epoch (i.e., per time window) adaptive re-optimization is mandatory \cite{Pachler2026}. Such re-optimization is not free, as it leads to relocating a running \glspl{vnf}, transferring its run-time state (e.g., \gls{ids}'s or \gls{siem}'s detection state, and encryption session keys) across the ISL mesh, and interrupting the traffic flow and momentarily degrading the security service the SFC chain is meant to provide \cite{Geng2024}.
A rotating topology forces a baseline of relocations that no objective can prevent. Hence, each re-optimization must penalize only avoidable migrations, in which the prior satellite remains visible and capacity-feasible. Yet, the solver relocates the \gls{vnf}s to improve the single-epoch objective. Forced migrations, where the prior satellite has dropped below the elevation threshold or has become capacity-infeasible, are tolerated. A myopic solver that treats all migrations uniformly over-penalizes necessary handoffs and under-penalizes controllable churn \cite{Mahyoub2026}.

On the other hand, when functions of different slices share the same \gls{vnf} instance in a satellite, confidential traffic and cryptographic material may be exposed to cross-slice side-channel attacks such as cache-timing~\cite{liu2015last} and co-resident memory probing~\cite{ristenpart2009hey}. Because the topology changes every epoch, this exposure shifts dynamically and must be mitigated by a security objective embedded within the placement optimization strategy.

Since VNF migration cost and co-location risk trade against each other, and co-location couples otherwise independent slices, the placement is chosen by a single joint optimization that minimizes both avoidable-migration cost and security risk over all slices. The main scalability challenge is due to the fact that the joint all-slice formulation introduces large binary co-location indicator variables, one for each ordered pair of (slice, \gls{sfc}'s \gls{vnf}) tuples sharing the same \gls{vnf} instance in a particular satellite. This quadratic growth in slice counts renders the joint solution computationally intractable at a practical scale, motivating a principled decomposition strategy.

Based on the above, four requirements should be taken into account: (i)
adaptation to dynamic \gls{leo} topology through epoch-by-epoch re-optimization, (ii) distinction between avoidable and forced migrations, (iii) quantifiable cross-slice co-location risk objective, and (iv) tractable decomposition as the slice count grows. Prior work partially addressed these requirements, falling into four threads.
First, terrestrial \gls{milp}/heuristic \gls{sfc} placement \cite{cohen2015near,eramo2017approach,rankothge2017optimizing,biallach2024vnf} established algorithmic baselines, but assumed static substrates without \gls{leo} topology dynamics or multi-tenant security objectives.
Second, \gls{vnf} and \gls{sfc} orchestration for \gls{leo} satellite networks \cite{wang2020sfc,huang2020service,he2024iotj,minardi2025jsac}, spanning \gls{ilp} decompositions \cite{Jia2020,Jia2021}, demand-triggered migration heuristics~\cite{Geng2024}, graph-attention and cooperative learning methods \cite{He2024,Doan2025}, and flexible resource allocation \cite{Ahsan2026}, handled epoch-driven topology change but did not formalize cross-slice co-location security risk as a quantifiable multi-slice objective, even when serving multiple slices, such as \cite{Ahsan2026,Liu2026}, thus overlooking inter-slice security coupling and the avoidable/forced migration distinction.
Third, \gls{sfc} orchestration in \gls{sagin} \cite{Xu2024,Petrosino2023,Qin2023,Zheng2025,Wang2026,Liu2026b} extended this concept to heterogeneous multi-segment architectures, but likewise lacked security-aware multi-slice objectives and the avoidable/forced distinction.
Fourth, security-aware VNF and \gls{sfc} embedding \cite{zhang2021security,ali2020security,JGao2022,Wang2024,Dubba2024} enforced isolation and protection through physical separation, dedicated backup paths, or threat-aware scheduling, but only in terrestrial contexts without \gls{leo}-specific dynamic co-location risk modelling.

To the best of our knowledge, no prior work fully satisfied the four requirements. To address this gap, we present here a novel security- and migration-aware \gls{sfc} placement framework for \gls{leo} satellite networks. The proposed framework formulates a joint multi-slice \gls{milp} with McCormick-linearized co-location risk and an avoidable-migration epigraph, then applies an \acrfull{admm}-inspired penalized per-slice best-response decomposition \cite{boyd2011distributed,Takapoui2020} that transfers the cross-slice co-location coupling to a linear per-slice penalty. Two coordinated schedules are proposed and analyzed, namely a sequential Gauss-Seidel variant (S-\gls{admm}) that refreshes a cross-slice occupancy estimate within each sweep, and a parallel variant (P-\gls{admm}) that solves all slices concurrently and restores coordination through a deterministic collision-repair pass.

The main contributions of the paper are as follows:
\begin{enumerate}
\item \textit{Security- and Migration-Aware Formulation:} We formulate security- and migration-aware \gls{sfc} placement in \gls{leo} constellations as a single-epoch multi-slice \gls{milp}, subject to per-satellite CPU capacity, \gls{isl} capacity, and per-user \gls{e2e} delay constraints.

\item \textit{Scalable Decomposition:} To overcome the quadratic growth of co-location variables that renders the joint \gls{milp} intractable, we develop an \gls{admm}-inspired penalized per-slice best-response decomposition that reduces the problem to separable per-slice subproblems. Since a naive separable penalty leaves the cross-slice coupling inactive, we further develop two coordination schedules that restore it: a sequential Gauss-Seidel occupancy refresh (S-\gls{admm}) and a parallel Jacobi schedule with deterministic collision repair (P-\gls{admm}).

\item \textit{Validation:} We evaluate the framework on a satellite Walker-Delta constellation across a range of slice counts. The proposed methods reduce co-location risk by orders of magnitude relative to security-unaware optimization, remain feasible within the per-epoch budget at slice counts, for which the monolithic security-aware \gls{milp} times out, and expose a controllable security-migration trade-off between the two decomposition schedules.
\end{enumerate}

The remainder of the paper is organized as follows.
Section~\ref{sec:related} provides an overview of the related work.
Section~\ref{sec:system_model} presents the system model and Section~\ref{sec:problem_formulation} formulates the security- and migration-aware placement problem.
Section~\ref{sec:admm} develops the proposed S-\gls{admm} and P-\gls{admm} decompositions.
Section~\ref{sec:setup} describes the simulation setup and  Section~\ref{sec:results} presents the experimental results. 
Finally, Section~\ref{sec:conclusion} concludes the paper.

\section{Related Work}
\label{sec:related}
\begin{table*}[!t]
\renewcommand{\arraystretch}{1.05}
\caption{Comparison of the Proposed Framework with Reviewed Works}
\label{tab:related}
\centering
\scriptsize
\setlength{\tabcolsep}{2.6pt}
\begin{tabular}{@{}l c c c c c c c c@{}}
\toprule
\multirow{2}{*}{\textbf{Work}} &
  \multirow{2}{*}{\textbf{Year}} &
  \multirow{2}{*}{\textbf{Network}} &
  \multirow{2}{*}{\begin{tabular}[c]{@{}c@{}}\textbf{Multi-}\\\textbf{slice}\end{tabular}} &
  \multirow{2}{*}{\begin{tabular}[c]{@{}c@{}}\textbf{Security}\\\textbf{model}\end{tabular}} &
  \multirow{2}{*}{\begin{tabular}[c]{@{}c@{}}\textbf{Migration}\\\textbf{handling}\end{tabular}} &
  \multirow{2}{*}{\begin{tabular}[c]{@{}c@{}}\textbf{Objective}\\\textbf{function}\end{tabular}} &
  \multirow{2}{*}{\begin{tabular}[c]{@{}c@{}}\textbf{Optimisation}\\\textbf{method}\end{tabular}} &
  \multirow{2}{*}{\begin{tabular}[c]{@{}c@{}}\textbf{Scalable}\\\textbf{decomposition}\end{tabular}} \\[8pt]
\midrule
Jia \textit{et al.}~\cite{Jia2020,Jia2021} & 2020  & \gls{leo}  Satellites      & \xmark  & \xmark   & \xmark  & Min. resource consumption          & \gls{ilp}      & Dantzig-Wolfe decomposition \\
Gao \textit{et al.}~\cite{XGao2021} & 2021  & \gls{leo}  Satellites     & \xmark  & \xmark   & \xmark  & Min. resource utilization        & INLP/Greedy   & \xmark \\
Gao \textit{et al.}~\cite{XGao2022} & 2022  & \gls{leo}  Satellites    & \xmark  & \xmark   & \xmark  & Max. network payoff       & Potential~game    & Decentralization \\
Petrosino \textit{et al.}~\cite{Petrosino2023} & 2023  & LEO Satellites        & \checkmark  & \xmark   & \xmark  & Min. delay     & Heuristic    & \xmark \\
Xia \textit{et al.}~\cite{Xia2024} & 2024  & \gls{leo}  Satellites    & \xmark  & \xmark   & \xmark  & Min.~average~delay        & Approx. alg. & \xmark \\
Geng \textit{et al.}~\cite{Geng2024} & 2024  & \gls{leo}  Satellites        & \checkmark  & \xmark   & \checkmark & Max. acceptance + Min. migration & Heuristic    & \xmark \\
He \textit{et al.}~\cite{He2024} & 2024  & \gls{leo}  Satellites        & \xmark  & \xmark   & \checkmark  & Max.~acceptance + load fairness  & INLP/GAT     & \xmark \\
Yan \textit{et al.}~\cite{Yan2025} & 2025  & \gls{leo}   Satellites     & \xmark  & \xmark   & \xmark  & Max. network profit            & \gls{ilp}          & \xmark \\
Doan \textit{et al.}~\cite{Doan2025} & 2025  & \gls{leo}  Satellites        & \checkmark  & \xmark   & \xmark  & Min.~average~delay        & MAQL      & Decentralized \\
Ahsan \textit{et al.}~\cite{Ahsan2026} & 2026  & LEO Satellites     & \checkmark & \xmark   & \xmark  & Max. served requests     & \gls{milp}/SCA     & \xmark \\
Mahyoub \textit{et al.}~\cite{Mahyoub2026} & 2026  & \gls{leo}  Satellites        & \checkmark  & \xmark& \checkmark & Min. handovers + migrations   & MINLP        & Hierarchical decomposition \\
Liu \textit{et al.}~\cite{Liu2026} & 2026  & \gls{leo}  Satellites        & \checkmark & \xmark   & \xmark  & Max. system utility          & DRL   & Clustering \\
Chen \textit{et al.}~\cite{Chen2026} & 2026  & \gls{leo}  Satellites        & \xmark  & \xmark   & \xmark  & Max. revenue + Min. latency         & \acrshort{rl}       & \xmark \\
Liu \textit{et al.}~\cite{Liu2026a} & 2026  & \gls{leo}  Satellites   & \checkmark  & \xmark   & \xmark  & Min. cost + latency + energy   & Heuristic   & \xmark \\
\addlinespace
Jia et al.~\cite{Jia2025} & 2025  & SAGIN      & \xmark  & \xmark   & \xmark  & Max. SFC deployment   & \gls{ilp}/DDQN     & \xmark \\
Liu \textit{et al.}~\cite{Liu2026b} & 2026  & SAGIN       & \checkmark  & \xmark   & \checkmark & Min. mig. + delay + jitter   & ORO/OLRO     & Graph aggregation \\
\addlinespace
Zhang \textit{et al.}~\cite{ZhangVNE} & 2021  & Terrestrial& \xmark  & Sec. Level& \xmark  & Max. long term return         & \acrshort{rl}    & \xmark \\
Gao \textit{et al.}~\cite{JGao2022} & 2022  & Terrestrial& \checkmark & Isolation& \xmark  & Min.~resources + Max. isolation     & INLP/Hypergraph   & \xmark \\
Wang \textit{et al.}~\cite{Wang2024} & 2024  & Terrestrial& \xmark  & Protection& \xmark  & Min. cost + Max. protection     & \gls{ilp}/Heuristic & \xmark \\
Dubba \textit{et al.}~\cite{Dubba2024} & 2024  & Terrestrial& \xmark  & Threat-Aware& \xmark  & Min. execution cost       & Heuristic    & \xmark \\[4pt]
\midrule
\textbf{This work} &
  \textbf{2026} &
  \textbf{LEO Satellites} &
  \textbf{\checkmark} &
  \checkmark  &
  \textbf{\checkmark} &
  \textbf{\begin{tabular}[c]{@{}c@{}}Min.~security risk + migrations\end{tabular}} &
  \textbf{\gls{milp}} &
  \textbf{\gls{admm}} \\
\bottomrule
\end{tabular}
\end{table*}
Integrating  \gls{sdn} and \gls{nfv} into \gls{leo} satellite networks raises a scheduling problem. Indeed, limited on-board resources must be allocated across a topology that restructures at every orbital epoch. Early studies cast the problem as an integer program over a time-evolving graph. For example, Jia \textit{et al.} formulated \gls{vnf} orchestration in software-defined \gls{leo} small-satellite networks as an \gls{ilp} targeting resource minimization and solved it through a Dantzig-Wolfe decomposition with column generation~\cite{Jia2020,Jia2021}. In contrast, Yan \textit{et al.} reused \gls{vnf} instances across concurrent flows to lower initialization delay and raise profit~\cite{Yan2025}. Building toward more heterogeneous service models, Ahsan \textit{et al.} generalized these single-service formulations to a multi-slice eMBB/mMTC \gls{milp} solved by successive convex approximation (SCA) with epoch-aware path reconfiguration \cite{Ahsan2026}.

Heuristic, temporal, and look-ahead studies instead track the structure imposed by orbital motion at a lower cost. For example, Petrosino \textit{et al.} allocated \glspl{sfc} among \gls{leo} CubeSats over a planning horizon under intermittent visibility~\cite{Petrosino2023}. Moreover, Geng \textit{et al.} deferred relocation until arriving requests can no longer be accommodated~\cite{Geng2024}. This demand-triggered migration heuristic is proven to approach the offline optimal cost. Liu \textit{et al.} jointly reduced \gls{sfc} migration frequency, delay, and jitter in SAGIN. Exploiting orbital regularity, Gao \textit{et al.} used trajectory predictions ~\cite{XGao2021} and later a potential-game formulation with a decentralized Nash-equilibrium solver~\cite{XGao2022} for location-aware placement, while Xia \textit{et al.} introduced the chaining-orbit concept for delay-aware service chaining under uncertain \gls{isl} delays~\cite{Xia2024}. Closest to the present work is our previous study~\cite{Mahyoub2026}, where we reduced handovers and security-function migrations through temporal stability regularization, warm-start seeding, and hierarchical decomposition of user association and \gls{sfc} placement.

Recently, learning-based studies have addressed the scale and non-stationarity of large constellations. For instance, He \textit{et al.} cast \gls{sfc} orchestration in dynamic \gls{leo} networks as an integer nonlinear program (INLP) and solved it with a graph attention network (GAT) that jointly maximizes service acceptance and load fairness~\cite{He2024}. Similarly, Doan \textit{et al.} addressed joint \gls{vnf} caching and placement through cooperative multi-agent Q-learning augmented by Bayesian optimization~\cite{Doan2025}. 
Alternatively, Jia \textit{et al.} combined an \gls{ilp} formulation with a double deep Q-network  (DDQN) over a reconfigurable time-expansion graph to resolve inter-task resource conflicts and urgency constraints in SAGIN~\cite{Jia2025}. 
Other works focused on the energy efficiency of the placement strategy. 
For example,  Chen \textit{et al.} proposed a battery-aware proximal-policy-optimization (PPO) framework for \gls{sfc} placement under fluctuating satellite energy levels ~\cite{Chen2026}, while Liu \textit{et al.} developed multi-objective heuristics for satellite optical networks that jointly balance battery lifetime, \gls{e2e} latency, and resource utilization under eclipse-driven power cycling~\cite{Liu2026a}.

The aforementioned studies have progressively broadened the \gls{sfc} and \gls{vnf} placement objective from resource efficiency toward service acceptance, latency, energy, and migration cost, and have evolved from exact solvers through heuristics to learning-based policies. Across all the above works, none formalized cross-slice co-location security risk as a quantifiable, optimizable objective grounded in established security standards.

Security-aware \gls{vnf} placement has grown in terrestrial \gls{sfc} contexts along two threads. The first enforces physical isolation as a combinatorial constraint. For instance, Gao \textit{et al.} proposed multi-tenant isolation as weighted hypergraph matching that forbids co-locating conflicting tenants~\cite{JGao2022}, while Wang \textit{et al.} reoriented this toward survivability with asymmetric dedicated protection across primary/backup chains~\cite{Wang2024}. The second relaxes hard isolation into learned security-level requirements. For example, Zhang \textit{et al.} encoded node-level security demands into the substrate features of an \gls{rl} policy~\cite{ZhangVNE}, and Dubba \textit{et al.} added an explicit adversary model with \gls{sla}-enforcing heuristics trading cost against security~\cite{Dubba2024}. However, these works operated in terrestrial or technology-agnostic settings and treated security as a static, per-node attribute. None of these works model co-location risk as a function of data sensitivity, slice criticality, and bilateral isolation policy grounded in ISO/NIST risk management principles, namely the ISO~31000 likelihood$\,\times\,$consequence risk structure~\cite{iso31000} and the NIST~SP~800-53 boundary- and process-isolation control families~\cite{nist800}. This distinction is decisive in the LEO setting, where co-location exposure is not a one-time design constraint but an epoch-by-epoch objective that shifts as the constellation rotates, and thus should be embedded within the adaptive placement optimization approach.

This work closes the above gaps by proposing the first multi-slice security-aware \gls{milp} for \gls{leo} \gls{sfc} placement, combining a formally grounded risk model with analytic bounds and an \gls{admm} decomposition that keeps the optimization tractable.
Table~\ref{tab:related} positions reviewed works against the contributions of the proposed framework.

\section{System Model}
\label{sec:system_model}

\subsection{Network Model}
\label{subsec:constellation}

We consider a Walker-Delta  \gls{leo} satellite network comprising $S$~satellites distributed across $P$~orbital planes, with $S/P$~satellites per plane, orbital altitude $h_{\text{orb}}$~km, and inclination $\iota$~degrees.
Let $\calS = \{1,\ldots,S\}$ denote the satellite index set.
Time is discretised into orchestration epochs indexed by $t \in \calT = \{0,1,\ldots,T-1\}$, each of duration $\Delta t$~seconds \cite{Wang2022}.
At each epoch~$t$, the network state is captured by a topology snapshot
\(
  \mathcal{G}^{(t)} = \bigl(\calS,\;\mathcal{E}^{(t)},\;\bm{d}^{(t)},\;\bm{c}^{(t)},\;
 \bm{p}^{(t)}\bigr), \)
where $\mathcal{E}^{(t)} \subseteq \calS \times \calS$ is the directed \gls{isl} graph, $\bm{d}^{(t)}$ are the corresponding propagation delays, $\bm{c}^{(t)}$ are per-satellite CPU capacities, and $\bm{p}^{(t)} = \{\bm{r}_s^{(t)}\}_{s \in \calS}$ are the satellite position vectors in $\mathbb{R}^3$. The \gls{isl} graph $\mathcal{E}^{(t)}$, the propagation delays $\bm{d}^{(t)}$, and user-satellite visibility are all derived from $\bm{p}^{(t)}$, which is recomputed in each epoch as the constellation evolves \cite{mahyoub2026_visibility}. Table~\ref{tab:notation} summarizes the complete notation used in this paper.

\begin{table}[!t]
  \renewcommand{\arraystretch}{1.2}
  \caption{Summary of Notation}
  \label{tab:notation}
  \centering
  \begin{tabular}{@{}cl@{}}
    \toprule
    \textbf{Symbol} & \textbf{Definition} \\
    \midrule
    $\calS$, $S$          & Satellite set and cardinality \\
    $\calN$, $N$          & Slice set and cardinality \\
    $\calF$               & VNF type set \\
    $\calI$, $I$          & Instance index set and cardinality \\
    $\calU_n$, $U_n$      & User set and size for slice $n$ \\
    $L_n$                 & Chain length of slice $n$ \\
    $f_l^n$               & VNF type at position $l$ of chain $n$ \\
    $R_f$       & Security sensitivity of function $f$ \\
    $C_n$       & Criticality of slice $n$ \\
    $\Phi_{n,n'}$ & Isolation policy coefficient \\
    $w_{n,n'}^{f}$          & Risk weight (Definition~\ref{def:risk_weight}) \\
    $D_f^{\text{mig}}$    & Migration disruption cost for function $f$ \\
    $O_{f,i,s}^{\text{actv}}$ & Activation CPU overhead of instance $(f,i,s)$ \\
    $O_{f,i,s}^{\text{incr}}$ & Per-user CPU cost of instance $(f,i,s)$ \\
    $O_{f,i,s}^{\text{proc}}$           & Processing delay of instance $(f,i,s)$ (ms) \\
    $\mathcal{C}_s^{\text{cpu}}$ & CPU capacity of satellite $s$ \\
    $\bar{T}_{n,u}$       & E2E delay budget for user $u$ in slice $n$ (ms) \\
    $d_{s,s'}^{\text{SP}}$ & Shortest-path ISL delay $s \to s'$ (ms) \\
    $d_{n,u,s}^{\text{acc}}$ & Access delay: user $u$ of slice $n$ to satellite $s$ (ms) \\
    $\Omega_{n,u}^{(t)}$  & Visible satellite set for user $u$ at epoch $t$ \\
    $\pi_{n,u}^{l}$ & Prior-assignment feasibility\\
    $F_{\max}$              & Maximum concurrent flows per ISL link \\
    $\rho$            & Proximal penalty weight \\
    $\alpha_{n,f,i,s}$         & Cross-slice occupancy estimate \\
    $\delta_{\mathrm{obj}}$ & Objective-stagnation tolerance \\
    $\beta_{n,u,l}^{i,s}$ & Placement decision variable \\
    $\gamma_{f,i,s}$   & VNF activation indicator \\
    $z_{n,l}^{i,s}$      & Slice-activation indicator \\
    $y_{n,n',l,l'}^{i,s}$ & Co-location linearisation variable \\
    $\mu_{n,u}^{l}$       & Avoidable migration indicator \\
    $v_{n,u,l}^{s,s'}$     & ISL-hop McCormick variable \\
    $\sigma_{n,u}^{l}$    &  stay indicator at prior assignment $(\hat\imath,\hat{s})$ \\
    $\xi_{n,u,l}^{s}$     &  aggregate satellite-assignment at position $l$ ($\sum_i \beta_{n,u,l}^{i,s}$) \\
    $\bar{U}$           &  maximum users per slice \\
    $L_{\max}$          &  maximum chain length \\
    \bottomrule
  \end{tabular}
\end{table}

At epoch~$t$, satellite~$s$ maintains \glspl{isl} to at most $k$~nearest neighbours (typically $k=4$: intra-plane fore/aft and inter-plane left/right \cite{Pachler2026}).
The direct-link propagation delay between adjacent satellites $s, s' \in \calS$ is $d_{s,s'}^{(t)} = \|\bm{r}_s^{(t)} - \bm{r}_{s'}^{(t)}\|_2 / c$, where 
$c$ is the speed of light \cite{mahyoub2026_visibility}. For non-adjacent satellite pairs $(s,s')\notin\mathcal{E}^{(t)}$, the satellite pairs' shortest-path (SP) delay
$d_{s,s'}^{\text{SP},(t)}$ is computed via Dijkstra's algorithm over $\mathcal{G}^{(t)}$, at every time epoch to reflect the evolving \gls{isl} graph.

Let $\calN = \{1,\ldots,N\}$ be the set of network slices, each representing an independent slice with a dedicated security service requirement.
The security requirements of a  slice $n \in \calN$ is characterised by an ordered \gls{sfc}
\(
  \calF_n = \{f_1^n, f_2^n, \ldots, f_{L_n}^n\},
\)
where $\calF$ is the set of admissible security \gls{vnf} types (e.g., $\calF = \{\text{FW}, \text{IDS}, \text{ENC}, \text{TM}, \text{SIEM}\}$), and $L_n = |\calF_n|$ is the chain length.
A slice~$n$ serves a set $\calU_n = \{1,\ldots,U_n\}$ of ground users, each with a known location (i.e., latitude and longitude). 
User~$u$ of slice~$n$ can communicate with satellite~$s$ at epoch~$t$ if and only if (iff) the elevation angle $\theta_{n,u,s}^{(t)}$ satisfies $\theta_{n,u,s}^{(t)} \geq \theta_{\min}$, where $\theta_{\min}$ is the minimum visibility elevation angle.
The visible satellite set for user~$u \in \calU_n$ at epoch~$t$ is defined as
\(
  \Omega_{n,u}^{(t)} = \bigl\{ s \in \calS |
    \theta_{n,u,s}^{(t)} \geq \theta_{\min} \bigr\}.
\)
Traffic from user $u$ must traverse the slice's \gls{sfc} in strict order $f_1^n \to f_2^n \to \cdots \to f_{L_n}^n$, with respect to its \gls{e2e} latency budget $\bar{T}_{n,u}$~ms.
For each function type $f \in \calF$ and satellite $s \in \calS$, there are $I$~candidate \gls{vnf} instances indexed by $\calI = \{1,\ldots,I\}$.
An instance $(f, i, s)$ (with $i \in \calI$) is characterised by different overheads, including fixed activation CPU overhead $O_{f,i,s}^{\text{actv}}$,  per-user incremental CPU load $O_{f,i,s}^{\text{incr}}$, and \gls{vnf} processing delay $O_{f,i,s}^{\text{proc}}$.

\subsection{Co-location Security Risk Model}
\label{subsec:risk_model}

The main security concern in multi-slice computing is \gls{vnf} co-location. When functions of different slices share the same physical instance on the same satellite, confidential traffic flows and cryptographic material may be exposed to cross-slice side-channel attacks~\cite{ristenpart2009hey,mahyoub_slicing_2025,AbdulGhaffar2024}.
We represent this risk via a multiplicative security model inspired by ISO~31000~\cite{iso31000} and NIST~SP~800-53~\cite{nist800}.
The ISO~31000 risk management vocabulary (likelihood $\times$ consequence, risk appetite, controls) informs the multiplicative structure of Definition~\ref{def:risk_weight} below.

\begin{definition}[Security Co-location  Risk Weight]
\label{def:risk_weight}
  The security co-location risk weight between slices $n, n' \in \calN$, $n < n'$, for function type $f \in \calF$ is defined as
  \begin{equation}
    w_{n,n'}^{f} = R_f \cdot \Phi_{n,n'} \cdot C_n \cdot C_{n'} \;\geq 0,
    \label{eq:risk_weight}
  \end{equation}
  where $R_f$ is the security sensitivity of function type~$f$ (higher for cryptographic or intrusion-detection functions), $C_n$ is the criticality of slice~$n$ (reflecting data classification level), and $\Phi_{n,n'}$ is the bilateral isolation policy coefficient between slices $n$ and $n'$, with $\Phi_{n,n'} = \Phi_{n',n}$ (symmetric).
  The multiplicative structure implies that any zero factor (zero sensitivity, zero criticality, or enforced isolation) collapses the full risk contribution to zero, consistent with zero-trust policy semantics.
\end{definition}

The NIST~SP~800-53  control families \cite{nist800}, relevant to multi-slice \gls{sfc} isolation, include: SC-2 (Application Partitioning), SC-7 (Boundary Protection), and SC-39 (Process Isolation). These controls map to the isolation policy coefficient $\Phi_{n,n'}$ and the activation of zero-trust enforcement ($\Phi_{n,n'} = 0$ corresponds to full SC-39 isolation between slices $n$ and $n'$).
Moreover, cross-slice isolation requirement on shared infrastructure is recognized across the 5G/NFV slicing-security literature ~\cite{mahyoub_slicing_2025,ali2020security}. 
Since every factor in~\eqref{eq:risk_weight} is symmetric in its slice indices, we extend the risk weight to all ordered pairs by setting $w_{n',n}^{f} \triangleq w_{n,n'}^{f}$ for $n' < n$.
\begin{definition}[Co-location Security Risk]
\label{def:risk_exact}
  The {co-location security risk} is calculated as
  \begin{align}
    \Risk &=
      \sum_{\substack{n,n' \in \calN \\ n < n'}}
      \sum_{f \in \calF_n \cap \calF_{n'}}
      \sum_{i \in \calI}
      \sum_{s \in \calS}
      w_{n,n'}^{f} ~
      \bigl|\mathcal{A}_{n,f}^{i,s}\bigr| \cdot 
      \bigl|\mathcal{A}_{n',f}^{i,s}\bigr|,
    \label{eq:risk_exact}
  \end{align}
  where $\mathcal{A}_{n,f}^{i,s}$ is the set of users of slice~$n$ assigned to instance $(f,i,s)$.
\end{definition}
We treat each pair of co-resident cross-slice flows on a shared instance as one independent side-channel exposure~\cite{ristenpart2009hey}. The per-instance term $|\mathcal{A}_{n,f}^{i,s}| \cdot |\mathcal{A}_{n',f}^{i,s}|$ in~\eqref{eq:risk_exact} counts these exposure pairs, generalizing the binary co-residency model standard in the side-channel literature. Following the ISO~31000 likelihood~$\times$~consequence structure, each contribution factorises into a consequence component $(R_f\,C_n C_{n'})$ and a \emph{likelihood} component $(\Phi_{n,n'}\,~|\mathcal{A}_{n,f}^{i,s}|~|\mathcal{A}_{n',f}^{i,s}|)$.

Direct minimization of $\Risk$ is intractable as it involves products of integer sums. Hence, we introduce the coarse co-location indicator
\begin{equation}
  z_{n,l}^{i,s} \in \{0,1\},\quad
  z_{n,l}^{i,s} = 1 \iff \exists\, u \in \calU_n |
    \beta_{n,u,l}^{i,s} = 1,
  \label{eq:z_def}
\end{equation}
indicating whether there is a user of slice $n$ that is served by instance $(f,i,s)$ or not, and where $\beta_{n,u,l}^{i,s}\!\in\!\{0,1\}$ is a binary variable equals~1 iff the $l$-th \gls{vnf} of $\calF_n$ is assigned to instance~$i$ on satellite~$s$ for a user $u$ belonging to slice $n$. For the coarse co-location indicator $ z_{n,l}^{i,s}$, we establish the following analytic bounds.

\begin{proposition}[Co-location Security Risk Bounds]
\label{prop:risk_bounds}
  For any assignment $\bfbeta=[\beta_{n,u,l}^{i,s}]_{n\in\calN,\;u\in\calU_n,\; l\in\{1,\ldots,L_n\},\;i\in\calI,\;s\in\calS}$, define the bilinear co-location indicator
  \begin{equation}
    y_{n,n',l,l'}^{i,s} = z_{n,l}^{i,s} \cdot z_{n',l'}^{i,s},
    \label{eq:y_bilinear}
  \end{equation}
  where $(n,l)$ and $(n',l')$ share the same function type ($f_l^n = f_{l'}^{n'} = f$).
  Then,
\(
    \Risk^{\text{LB}} \;\leq\; \Risk \;\leq\; \Risk^{\text{UB}},
    \label{eq:risk_bounds}
\)
  where
  \begin{align}
    \Risk^{\text{LB}} &=
      \sum_{\substack{n,n' \in \calN,\; n < n' \\l \in [L_n],\; l' \in [L_{n'}] }}
      \sum_{i \in \calI,\; s \in \calS}
      w_{n,n'}^{f}\, ~ y_{n,n',l,l'}^{i,s},
      \label{eq:risk_lb} \\[4pt]
    \Risk^{\text{UB}} &=
      \sum_{\substack{n,n' \in \calN,\; n < n' \\l \in [L_n],\; l' \in [L_{n'}] }}
      \sum_{i \in \calI,\; s \in \calS}
      w_{n,n'}^{f}\, ~ |\calU_n|\cdot |\calU_{n'}|\,  ~ y_{n,n',l,l'}^{i,s}.
      \label{eq:risk_ub}
  \end{align}
  \end{proposition}
  \begin{proof}[Proof sketch]
  When $y_{n,n',l,l'}^{i,s}=1$, at least one user in $\calU_n$ and at least one user in $\calU_{n'}$ use instance $(f,i,s)$, such that
  $|\mathcal{A}_{n,f}^{i,s}| \cdot |\mathcal{A}_{n',f}^{i,s}| \geq 1 = y_{n,n',l,l'}^{i,s}$, establishing the lower bound.
  Since $|\mathcal{A}_{n,f}^{i,s}| \leq |\calU_n|$ and $|\mathcal{A}_{n',f}^{i,s}| \leq |\calU_{n'}|$,  their product is at most $|\calU_n|\cdot |\calU_{n'}|$, establishing the upper bound. When $y_{n,n',l,l'}^{i,s}=0$, at least one of the sets is empty, so the product is zero and all bounds are tight.
 \end{proof}

Note that $\Risk^{\text{LB}}$ in \eqref{eq:risk_lb} can be seen as the binary co-residency risk standard in the side-channel attacks. It charges $w_{n,n'}^{f}$ once per co-resident slice pair, independent of workload. 
The exact risk~\eqref{eq:risk_exact} and worst-case bound~\eqref{eq:risk_ub} satisfy $\Risk^{\text{LB}} \le \Risk \le \Risk^{\text{UB}}$ and serve as evaluation metrics. Since $w_{n,n'}^{f}\ge 0$, $\Risk^{\text{LB}}=0$ iff $\Risk=0$ (Definition~\ref{def:risk_exact}). Away from zero, the per-slice consolidation provides $\Risk = U^{2}\,\Risk^{\text{LB}}$, thus, $\Risk^{\text{LB}}$ is an order-preserving constant-factor surrogate for $\Risk$.

\section{Problem Formulation}
\label{sec:problem_formulation}
This section formalizes the single-epoch placement problem, defines the binary decision variables, and states the feasibility constraints. For readability, the epoch index $t$ is omitted from the per-epoch variables and constants throughout the paper.

\subsection{Decision Variables}
\label{subsec:variables}
The decision variables in our system are as follows:
\begin{itemize}
\item $\beta_{n,u,l}^{i,s} \in \{0,1\}:$ \gls{vnf}-usage indicator, equals~1 iff the $l$-th \gls{vnf} of $\calF_n$ is assigned to instance~$i$ on satellite~$s$ for a user $u$ belonging to slice $n$. We refer to the triple $(n,u,l)$ as a chain placement,
\item $\gamma_{f,i,s} \in \{0,1\}$:   \gls{vnf} activation indicator, equal to~1 iff at least one user is served by instance $(f,i,s)$, 
\item $z_{n,l}^{i,s}\!\in\!\{0,1\}$:   instance-usage indicator, equal to~1 iff any user of slice~$n$ uses instance $(f_l^n,i,s)$, 
\item $\mu_{n,u}^{l}\!\in\!\{0,1\}$: avoidable-migration indicator, equal to~1 iff the prior-epoch assignment remained feasible, but the solver chose otherwise.
\end{itemize}

\subsection{Constraint Set}
\label{subsec:constraints}
A feasible placement must satisfy the following six constraint families, governing user-to-instance assignment (C1), instance activation (C2), satellite CPU capacity (C3), avoidable-migration accounting (C4), \gls{e2e} delay (C5), and \gls{isl} link capacity (C6).

\noindent\textbf{(C1) Unique assignment:}
Each user $u \in \calU_n$ must be assigned to exactly one instance at each chain position~$l$ of $\calF_n$:
\begin{equation}
  \sum_{i \in \calI}\sum_{s \in \calS}
    \beta_{n,u,l}^{i,s} = 1,
  \quad \forall\, n \in \calN,\; u \in \calU_n,\; l \in \{1,\ldots,L_n\}.
  \tag{C1}\label{C1}
\end{equation}
When $l = 1$, the summation over $s$ is restricted to the visible satellite set $\Omega_{n,u}^{(t)}$. The \gls{sfc}'s \glspl{vnf} at positions $l>1$ may therefore reside on any satellite in $\calS$, including satellites not visible to the user, and are reached via the SP \gls{isl} route. 

\noindent\textbf{(C2) Activation coupling:}
A \gls{vnf} instance $(f,i,s)$ is activated only when used by at least one user.  This \gls{vnf} activation coupling is enforced by:
\begin{align}
  \beta_{n,u,l}^{i,s} &\leq \gamma_{f_l^n,\,i,\,s},
  \; \forall\, n \in \calN, u \in \calU_n, l\in \{1,\ldots,L_n\}, i \in \calI,\nonumber \\ s \in \calS,
    \tag{C2a}\label{C2a} \\
  \gamma_{f,i,s} &\leq
    \sum_{\substack{n \in \calN,\; u \in \calU_n \\
    \{l|\, f_l^n = f\}}}
    \beta_{n,u,l}^{i,s},
    \quad \forall\, f \in \calF,\; i \in \calI,\; s \in \calS.
    \tag{C2b}\label{C2b}
\end{align}

\noindent\textbf{(C3) Satellite CPU capacity:} The total CPU load on satellite~$s$  is
\vspace{-.2cm}
\begin{align}
  \text{CPU}_s
    &= \sum_{f \in \calF}\sum_{i \in \calI}
        O_{f,i,s}^{\text{actv}}\, ~ \gamma_{f,i,s}\nonumber \\
    &+ \sum_{n \in \calN}\sum_{u \in \calU_n}\sum_{l=1}^{L_n}\sum_{i \in \calI}
        O_{f_l^n,i,s}^{\text{incr}}\, ~ \beta_{n,u,l}^{i,s},
  \label{eq:cpu_load}
\end{align}
Consequently, the aggregate CPU load on each satellite must not exceed its capacity:
\begin{align}
\text{CPU}_s
  \leq \mathcal{C}_s^{\text{cpu}},
  \quad \forall\, s \in \calS.
  \tag{C3}\label{C3}
\end{align}

\noindent\textbf{(C4) Avoidable migration epigraph:} Let \text{Mig} denote the total migration disruption cost calculated as
\begin{equation}
  \text{Mig}
    = \sum_{n \in \calN}\sum_{u \in \calU_n}\sum_{l=1}^{L_n}
      D_{f_l^n}^{\text{mig}}\, ~ \mu_{n,u}^{l}, 
  \label{eq:migcost}
\end{equation}
where $D_{f_l^n}^{\text{mig}}$ is the migration disruption cost for the security function $f$.
Let $(\hat{\imath},\hat{s})$ denote the prior-epoch assignment of chain placement $(n,u,l)$ and $\pi_{n,u}^{l} \in \{0,1\}$ be the precomputed feasibility flag, i.e.,  $\pi_{n,u}^{l} = 1$ iff satellite~$\hat{s}$ is visible to user~$u$ at the current epoch and the residual capacity of~$\hat{s}$ can accommodate the user.
We define the stay indicator as the shorthand $\sigma_{n,u}^{l} \triangleq \beta_{n,u,l}^{\hat{\imath},\hat{s}} \in \{0,1\}$.
The avoidable migration variable satisfies
\(
  \mu_{n,u}^{l} = \max\bigl(0,\; \pi_{n,u}^{l} - \sigma_{n,u}^{l}\bigr),
  \label{eq:mu_prop}
\)
which is linearised as the three-constraint epigraph:
\begin{align}
  \mu_{n,u}^{l} &\geq \pi_{n,u}^{l} - \sigma_{n,u}^{l},
    \tag{C4a}\label{C4a} \\
  \mu_{n,u}^{l} &\leq \pi_{n,u}^{l},
    \tag{C4b}\label{C4b} \\
  \mu_{n,u}^{l} &\leq 1 - \sigma_{n,u}^{l}.
    \tag{C4c}\label{C4c}
\end{align}
When $\pi_{n,u}^{l} = 0$ (prior assignment infeasible), \eqref{C4b}--\eqref{C4c} force $\mu_{n,u}^{l}=0$, excluding involuntary handovers from the penalty. The flag $\pi_{n,u}^{l}$ is necessarily precomputed from epoch-$(t{-}1)$ information. Indeed, the current-epoch cross-slice load on $\hat{s}$ is unknown until the joint placement is solved, so letting $\pi$ depend on it would be circular. Hence, it estimates the residual capacity of $\hat{s}$ from the prior committed load. The \gls{milp} still enforces the exact capacity constraint~\eqref{C3}, such that every returned placement is feasible, while only the avoidable-versus-forced classification is approximate. If current-epoch contention forces a user off a prior satellite that appeared feasible under the $(t{-}1)$ estimate, the migration is deemed avoidable. Hence, the approximation is conservative as it can only inflate the proposed methods' reported avoidable-migration count, but never deflate it. 

\noindent\textbf{(C5) \gls{e2e} delay budget:}
The total E2E latency of each user must not exceed its budget $\bar{T}_{n,u}$, i.e.,
\begin{multline}
  \underbrace{\sum_{i\in \calI,s \in \Omega_{n,u}^{(t)}}
    d_{n,u,s}^{\text{acc}}\, ~ \beta_{n,u,1}^{i,s}}_{\text{access delay}}
  + \underbrace{\sum_{l=1}^{L_n}\sum_{i \in \calI,s \in \calS}
 O_{f_l^n,i,s}^{\text{proc}}\, ~ \beta_{n,u,l}^{i,s}}_{\text{processing delays}} \\
  + \underbrace{\sum_{l=1}^{L_n-1}\sum_{s \in \calS, s \neq s'}
    d_{s,s'}^{\text{SP}}\, ~ v_{n,u,l}^{s,s'}}_{\text{ISL propagation}}
  \leq \bar{T}_{n,u},  \;\forall\, n \in \calN, u \in \calU_n, 
  \tag{C5}\label{C5}
\end{multline}
where the McCormick variable $v_{n,u,l}^{s,s'} \in [0,1]$ linearizes the ISL routing product (see \eqref{eq:h_lb}--\eqref{eq:h_ub2}  below) and it is created only for reachable $(s,s')$ pairs to limit the model size.

\noindent\textbf{(C6) \gls{isl} link capacity:}
The aggregate number of user flows routed over each directed \gls{isl} must not exceed the link's flow capacity~$F_{\max}$, i.e.,
\begin{equation}
  \sum_{n \in \calN}\sum_{u \in \calU_n}\sum_{l=1}^{L_n-1}
    v_{n,u,l}^{s,s'}
  \leq F_{\max},
  \quad \forall\,(s,s') \in \mathcal{E}^{(t)}.
  \tag{C6}\label{C6}
\end{equation}

\subsection{Linearisation of Bilinear Terms}
\label{subsec:linearisation}

The above formulations contain two classes of bilinear products that must be linearized for \gls{milp} tractability. The normalization is discussed in the following two subsections.

\subsubsection{Co-location Security Risk Linearisation}
The indicator $y_{n,n',l,l'}^{i,s} = z_{n,l}^{i,s} \cdot  z_{n',l'}^{i,s}$ in \eqref{eq:risk_lb} is the product of two binary variables. Since both factors are binary, the standard McCormick~\cite{mccormick1976computability} linearisation reduces to the three-constraint LP relaxation:
\begin{align}
  y_{n,n',l,l'}^{i,s} &\leq z_{n,l}^{i,s},
    \label{eq:y_mcc1} \\
  y_{n,n',l,l'}^{i,s} &\leq z_{n',l'}^{i,s},
    \label{eq:y_mcc2} \\
  y_{n,n',l,l'}^{i,s} &\geq z_{n,l}^{i,s} + z_{n',l'}^{i,s} - 1.
    \label{eq:y_mcc3}
\end{align}
The slice-activation indicator $z_{n,l}^{i,s}$ is related to
$\bfbeta$ via
\begin{align}
  z_{n,l}^{i,s} \geq \beta_{n,u,l}^{i,s},
    \text{ and } 
  z_{n,l}^{i,s} \leq \sum_{u \in \calU_n} \beta_{n,u,l}^{i,s}.
    \label{eq:z_ub}
\end{align}
Since $\beta_{n,u,l}^{i,s}$ and $z_{n,l}^{i,s}$ are binary and C1 guarantees $\sum_{i\in \calI,s\in \calS}\beta_{n,u,l}^{i,s}=1$, the relaxation of the McCormick envelope is exact. At any integer-feasible solution, $y_{n,n',l,l'}^{i,s}$ take binary values and are exactly equal to $z_{n,l}^{i,s} \cdot z_{n',l'}^{i,s}$.

\subsubsection{ISL Delay Linearization}
The \gls{isl} propagation term in \eqref{C5} depends on the pair of satellites hosting two consecutive \glspl{vnf}. Let \(\xi_{n,u,l}^{s} \triangleq \sum_{i \in \calI} \beta_{n,u,l}^{i,s} \in \{0,1\}
\)
be the aggregate satellite-assignment indicator, i.e., equal to~1 iff the $l$-th \gls{vnf} of user~$u$ is placed on satellite~$s$, regardless of the instance. Then, the product $\xi_{n,u,l}^{s} \cdot \xi_{n,u,l+1}^{s'}$ 
equals~1 when user~$u$ routes \gls{vnf}~$l$ through~$s$ and \gls{vnf}~$(l+1)$ through~$s'$ (i.e., when the flow traverses the directed \gls{isl} $(s,s')$ between hops $l$ and $l+1$).
This conjunction of two consecutive placement decisions is irreducibly bilinear. Therefore, we keep the aggregated indicator and introduce a single auxiliary variable $v_{n,u,l}^{s,s'} \in [0,1]$ per reachable hop, linearized by the binary-product (McCormick) envelope as follows:
\begin{align}
  v_{n,u,l}^{s,s'} \geq \xi_{n,u,l}^{s} + \xi_{n,u,l+1}^{s'} - 1,
    \label{eq:h_lb} \\
  v_{n,u,l}^{s,s'} \leq \xi_{n,u,l}^{s},
    \text{ and }
  v_{n,u,l}^{s,s'} \leq \xi_{n,u,l+1}^{s'},
    \label{eq:h_ub2}
\end{align}
for all chain placements $(n,u,l)$ and all satellite pairs $(s,s') \in \calS^2$ with $s \neq s'$ and $d_{s,s'}^{\text{SP}} > 0$.
Since C1 imposes $\xi_{n,u,l}^{s} \in \{0,1\}$, the envelope is exact at every integer-feasible point, as \eqref{eq:h_ub2} pins $v_{n,u,l}^{s,s'}$ to~$0$ unless both hops are active, and \eqref{eq:h_lb} forces it to~$1$ when they are active. Hence, $v_{n,u,l}^{s,s'} = \xi_{n,u,l}^{s}\,~\xi_{n,u,l+1}^{s'}$ holds.

\subsection{Objective Function}
\label{subsec:milp_complete}

The main objective is to minimize the weighted co-location security risk and \gls{vnf} migration cost, as follows:
\begin{align}
  \min_{\bfbeta,\bfgamma,\mathbf{z},\mathbf{y},\bm{\mu}}
  \quad&
  \omega_{\text{risk}}\,\frac{\Risk^{\text{LB}}}{\bar{R}}
  + \omega_{\text{mig}}\, \frac{\text{Mig}}{\bar{M}}
  \label{eq:milp_obj}, \\[2pt]
  \text{s.t.} \quad
  & \eqref{C1}\text{--}\eqref{C6},  
   \eqref{eq:y_mcc1}\text{--}\eqref{eq:h_ub2}.\nonumber
\end{align}
The selected values of the weighting vector $(\omega_{\text{risk}}, \omega_{\text{mig}})$ reflect the importance of each term in multi-slice \gls{leo} deployments.
The normalization constants $\bar{R}$ and $\bar{M}$ are computed as analytic upper bounds in a preprocessing step, so that each normalized term lies in $[0,1]$, making the objective a convex combination that avoids scale-induced bias.
The dominant variables count is from the $y$-variables (co-location indicators), whose number is $O(N^2 L_{\max}^2 IS)$ in the joint formulation.
This quadratic scaling motivates the \gls{admm} decomposition developed in the section below.

\section{Proposed \gls{admm}-inspired Decomposition Framework}
\label{sec:admm}

\subsection{\gls{admm}-inspired Decomposition Principle}
\label{subsec:admm_decomp}

For large constellations and several slices, the joint \gls{milp} formulated in \eqref{eq:milp_obj} is computationally intractable due to the $O(N^2 L_{\max}^2 IS)$ co-location variables and the joint optimization over $N ~ \sum_{n \in \calN} U_n$ users' combinations.
We exploit the spatial separability of the problem, where the per-slice placement subproblems become independent when the activations of the other slices on each instance are fixed.
Hence, we replace the joint solve with a penalized per-slice best response, i.e., an \gls{admm}-inspired block decomposition~\cite{boyd2011distributed}, similar to what is used for scalable structured optimization \cite{yu2021distributed,Wang2021,Asheralieva2024}. In this work, \gls{admm} is customized to multi-slice binary \gls{sfc} placement with cross-slice co-location coupling.
Let $\gamma_{n,f,i,s} \in \{0,1\}$ denote the per-slice activation of instance $(f,i,s)$ by slice~$n$ (equivalently $\max_{\{l|\,f_l^n = f\}} z_{n,l}^{i,s}$), the slice-local counterpart of the global activation indicator $\gamma_{f,i,s}$ in~\eqref{eq:cpu_load}.
The cross-slice coupling is carried by an occupancy estimate $\alpha_{n,f,i,s} \in [0,1]$ of the other slices' activations on instance $(f,i,s)$, maintained across \gls{admm} sweeps and initialized to $\alpha^{(0)} = 0$.

Let 
$\bfbeta^{(t)}$ denote the committed placement at epoch~$t$ and $\bfbeta^{(t-1)}$ the previous-epoch placement used as the fixed migration reference. Also, let $\tbeta^{(k)}$ and $\hbeta$ denote the merged and repaired placements, respectively, sharing the same structure.
For each slice~$n$, let $\bfbeta_n$, $\bfgamma_n$, and $\bm{\mu}_n$ denote the restrictions of the placement, activation, and avoidable-migration variables to slice~$n$ (i.e., the entries   $\beta_{n,u,l}^{i,s}$, $\gamma_{n,f,i,s}$, and $\mu_{n,u}^{l}$ that carry slice index~$n$). Moreover, let a superscript $^{(k)}$ marks the value of these blocks at ADMM sweep~$k$, such that $\bfbeta_n^{(k)}$ is slice~$n$'s placement after the sweep-$k$ solve. At each \gls{admm} sweep~$k$, the \emph{x-step} solves $N$ per-slice \glspl{milp} as follows:
\begin{align}
  \min_{\bfbeta_n, \bfgamma_n, \bm{\mu}_n}
  \quad
  & \omega_{\text{mig}}~ \frac{\,\text{Mig}_n}{\bar{M}}
  + \sum_{f,i,s}
      \phi_{n,f,i,s}^{(k)}\,~\gamma_{n,f,i,s},
  \label{eq:admm_xstep} \\
  \text{s.t.} \quad
  & \eqref{C1},\;\eqref{C2a}\text{--}\eqref{C2b},\;
    \eqref{eq:c3rem},\;
    \eqref{C4a}\text{--}\eqref{C4c},\;\eqref{C5}, \nonumber
\end{align}
where the residual-capacity constraint \eqref{eq:c3rem} applies the joint CPU bound of constraint ~\textbf{\eqref{C3}} to slice~$n$ only, against the capacity left by the other slices:
\begin{equation}
  \text{CPU}_{n,s} \;\le\; \mathcal{C}_s^{\text{cpu}} - \Lambda_{n,s}^{(k)},
  \quad \forall\, s \in \calS,
  \tag{C3$^{\text{rem}}_n$}\label{eq:c3rem}
\end{equation}
where $\text{CPU}_{n,s}$ is slice~$n$'s CPU load on satellite~$s$ and $\Lambda_{n,s}^{(k)}$ denotes the CPU load  reserved by the other slices. $\Lambda_{n,s}^{(k)}$ is set differently by the two coordination schedules (S-\gls{admm}) and (P-\gls{admm}) developed in Subsections~\ref{subsec:sadmm} and~\ref{subsec:padmm}, respectively, which differ only in which other slices' loads are reserved. Fig. \ref{fig:framework} illustrates the E2E overview of the proposed framework along with the decomposition strategy.
\begin{figure*}[!t]
  \centering
  \resizebox{\textwidth}{!}{%
  \begin{tikzpicture}[
    font=\footnotesize,
    >={Latex[length=2.2mm]},
    box/.style={rectangle, rounded corners=2pt, draw=black, thick,
                align=center, inner sep=3pt, font=\scriptsize},
    io/.style={box, fill=cIn, text width=3.0cm, minimum height=0.95cm},
    store/.style={cylinder, shape border rotate=90, aspect=0.22, draw=black,
                  thick, fill=cOut, align=center, font=\scriptsize,
                  minimum width=1.7cm, minimum height=1.2cm, inner sep=1pt},
    comp/.style={box, minimum height=1.0cm},
    flow/.style={->, thick},
    fb/.style={->, thick, dashed},
    coord/.style={->, semithick, densely dotted},
    junc/.style={circle, fill=black, inner sep=1.3pt},
    badge/.style={circle, draw=cRiskB, fill=cRisk, line width=0.7pt,
                  inner sep=0.5pt, minimum size=3.4mm, font=\bfseries\tiny, text=black}
  ]
  \node[io] (topo) at (1.6,3.2)
       {Topology snapshot $\mathcal{G}^{(t)}$\\[1pt] ISL graph, Dijkstra SP, visibility $\Omega^{(t)}$};
  \node[io] (slice) at (1.6,1.8)
       {Slices \& SFCs \& Users};
  \node[io, fill=cRisk, draw=cRiskB] (risk) at (1.6,0.5)
       {Co-location risk model\\ $\Rightarrow$ weights $w^{f}_{n,n'}$};
  \node[store] (prev) at (2.35,-1) {$\bm{\beta}^{(t-1)},\ \pi$\\ avoid./forced};
  \node[comp, fill=cSub, text width=2.7cm] (sub) at (6.9,1.4)
       {$N$ per-slice MILP $x$-steps\\ (each $O(N)$ variables)};
  \begin{scope}[on background layer]
    \node[comp, fill=cSub, text width=2.7cm, minimum height=1.0cm] at ($(sub)+(0.12,0.12)$){};
    \node[comp, fill=cSub, text width=2.7cm, minimum height=1.0cm] at ($(sub)+(0.24,0.24)$){};
  \end{scope}
  \node[comp, fill=cSub, text width=2.7cm] (pen) at (6.9,3.5)
       {Linear penalty $\phi_n^{(k)}$:\\ occupancy coupling\\ $+\,\rho(\tfrac12-\gamma^{(k-1)})$ anchor};
  \node[comp, fill=cS, draw=black, text width=2.55cm] (sadmm) at (10.6,2.45)
       {\textbf{S-ADMM} (Gauss--Seidel)\\ Occupancy refresh};
  \node[comp, fill=cP, draw=black, dashed, text width=2.55cm] (padmm) at (10.6,0.35)
       {\textbf{P-ADMM} (Jacobi, parallel)};
  \node[comp, fill=cMerge, text width=2.0cm] (merge) at (13.9,1.4)
       {Merge \&\\ incumbent\\ tracking};
  \begin{scope}[on background layer]
    \node[draw=black, line width=1pt, rounded corners=4pt, fill=cEng,
          fit=(pen)(sub)(sadmm)(padmm)(merge), inner xsep=11pt, inner ysep=9pt] (engine) {};
  \end{scope}
  \node[anchor=west, font=\scriptsize\bfseries, text=white, fill=cTitle,
        rounded corners=2pt, inner sep=3pt] (etitle) at ($(engine.north west)+(0.15,0.30)$)
        {Penalized Per-Slice Best-Response ADMM  Decomposition};
  \node[box, fill=cRef, draw=black, dashed, text width=7.4cm] (joint) at (8.4,5.95)
       {\textbf{Joint multi-slice MILP}\\[2pt]
        $O(N^{2}L_{\max}^{2}IS)$ co-location variables $\Rightarrow$ intractable for large $N$};
  \draw[flow] (joint.south) -- node[right,font=\scriptsize,align=left,xshift=2pt]
       {decompose: replace $O(N^{2})$ coupling\\ with $N$ subproblems $+$ linear penalty}
       (joint.south |- etitle.north);
  \node[store, right=1cm of merge, fill=cOut] (out) {$\bm{\beta}^{(t)}$\\ $(\beta,\gamma,\mu)$};
  \node[comp, fill=cP, draw=black, dashed, text width=.9cm, font=\scriptsize]
        (repair) at ($(merge)!0.3!(out)+(0,-1.3)$)
        {Collision repair};
  \node[junc] (bus) at (5.0,1.4) {};
  \draw[flow] (topo.east)  -- ++(0.5,0) |- (bus);
  \draw[flow] (slice.east) -- ++(0.9,0) -- (bus);
  \draw[flow] (prev.east)  -- ++(0.5,0) |- (bus);
  \draw[flow] (bus) -- (sub.west);
  \draw[flow] (risk.east) -- ++(0.9,0) |- (pen.west);
  \draw[flow] (pen) -- (sub);
  \draw[flow] (sub.east) -- ++(0.45,0) |- (sadmm.west);
  \draw[flow] (sub.east) -- ++(0.45,0) |- (padmm.west);
  \draw[flow] (sadmm.east) -| (merge.north);
  \draw[flow] (padmm.east) -| (merge.south);
  \draw[flow] (merge) -- (out);
  \draw[flow] (merge.south) |- (repair.west);
  \draw[flow] (repair.east) -| (out.south west);
  \draw[fb] (out.south) |- (2.35,-2) -- (prev.south);
  \node[font=\scriptsize, anchor=south] at (8.4,-2)
       {$t\!\leftarrow\!t\!+\!1$:\ topology rotates;\quad
        $\bm{\beta}^{(t)}\!\to$ migration reference $\bm{\beta}^{(t-1)}$};
  \end{tikzpicture}}
  \caption{Overview of the proposed security- and migration-aware \gls{sfc} placement framework with \gls{admm} decomposition.
  }
  \label{fig:framework}
\end{figure*}
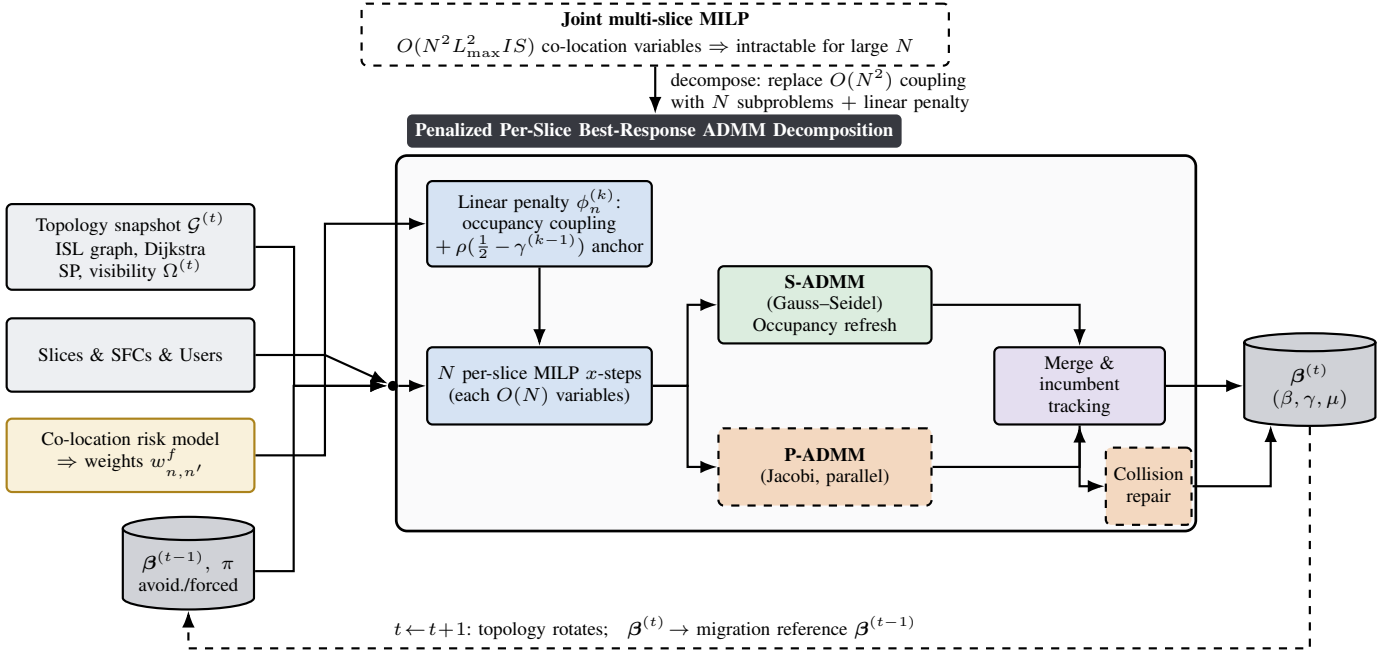 

The sequential schedule (S-\gls{admm}) reserves the loads of slices already solved in the current sweep ($n' < n$), whereas the parallel schedule (P-\gls{admm}) reserves all other slices' loads from the previous sweep ($n' \neq n$). Accordingly,
\begin{equation}
  \Lambda_{n,s}^{(k)} =
  \begin{cases}
    \displaystyle\sum_{n' < n} \ell_{n',s}^{(k)}
      & \text{(S-\gls{admm}),}\\[6pt]
    \displaystyle\sum_{n' \neq n} \ell_{n',s}^{(k-1)} = L_s^{(k-1)} - \ell_{n,s}^{(k-1)}
      & \text{(P-\gls{admm}),}
  \end{cases}
  \label{eq:residual_reserve}
\end{equation}
with $\ell_{n,s}^{(k)} = \text{CPU}_{n,s}$ the load slice~$n$ commits on satellite~$s$ in ADMM sweep~$k$ and $L_s^{(k)} = \sum_{n\in \calN} \ell_{n,s}^{(k)}$. Reserving from the previous sweep is what permits P-\gls{admm}'s fully parallel solves.

The co-location $y$-variables are absent from each per-slice subproblem and their cross-slice coupling is transferred entirely to the linear penalty $\phi_{n,f,i,s}^{(k)}$ on the activation $\gamma_{n,f,i,s}$, which linearizes the joint risk of the current occupancy estimates $\{\alpha_{n',f,i,s}\}_{n' \neq n}$ as follows:
\begin{equation}
  \phi_{n,f,i,s}^{(k)} =
    \frac{\omega_{\text{risk}}}{\bar{R}}
    \sum_{\substack{n'n, \in \calN \\n' \neq n}} w_{n,n'}^{f}\, ~ \alpha_{n',f,i,s}
    + \rho \Bigl(\tfrac{1}{2} - \gamma_{n,f,i,s}^{(k-1)}\Bigr),
  \label{eq:admm_penalty}
\end{equation}
where the first term is the linearized cross-slice co-location risk evaluated at the other slices' occupancy, and the second is a proximal anchor of strength $\rho > 0$ to slice~$n$'s own activation from the preceding sweep, $\gamma_{n,f,i,s}^{(k-1)}$ (with $\gamma^{(0)} = 0$), which damps inter-sweep oscillation. Instead of the dual-ascent step of consensus \gls{admm}, this decomposition maintains no dual variable and the coordination is carried entirely by the primal occupancy estimate $\alpha_{n',f,i,s}$ and the primal proximal anchor, such that $\phi_{n,f,i,s}^{(k)}$ is a linear coefficient on $\gamma_{n,f,i,s}$. 
Therefore, our approach should be viewed as a \emph{penalized best-response} \gls{admm}-inspired decomposition rather than a canonical consensus-\gls{admm} formulation. Correspondingly, our convergence result is the best-response fixed-point guarantee of Proposition~\ref{prop:convergence}, rather than the primal--dual convergence guarantee associated with conventional \gls{admm}. 
The second term in \eqref{eq:admm_penalty} is the exact linear representation of the squared proximal term $\tfrac{\rho}{2}\lVert\gamma_{n,f,i,s}-\gamma_{n,f,i,s}^{(k-1)}\rVert^2$ under binary activation variables. Since $\gamma \in \{0,1\}$ implies $\gamma^2 = \gamma$, the quadratic  term can be rewritten as $\rho(\tfrac{1}{2}-\gamma_{n,f,i,s}^{(k-1)})\,~\gamma_{n,f,i,s}$ plus a constant in $\gamma_{n,f,i,s}$. Consequently, the proximal anchor remains the previous activation state $\gamma_{n,f,i,s}^{(k-1)}$, as in standard proximal methods, while the coefficient $\tfrac{1}{2}$ arises solely from the binary nature of the decision variable and does not constitute an additional tuning parameter. When $\gamma_{n,f,i,s}^{(k-1)} = 0$, the second term becomes $+\rho/2$ interpreted as the cost of activating an instance and thus encourages retention of the previous inactive state. 
Conversely, when $\gamma_{n,f,i,s}^{(k-1)} = 1$,  the second term becomes $-\rho/2$, providing an equal incentive to preserve the previous active state.  These two cases are symmetric manifestations of the same proximal regularization mechanism, whose influence may be outweighed whenever the risk or migration terms yield a sufficiently large reduction in the per-slice objective. 

After each x-step, the occupancy estimate is updated to match the realized activation:
\begin{equation}
  \alpha_{n,f,i,s} \leftarrow \gamma_{n,f,i,s}^{(k)}.
  \label{eq:admm_occ}
\end{equation}
This update is performed either immediately after each slice optimization within the sweep (S-\gls{admm}, Subsection~\ref{subsec:sadmm}) or collectively once all slice subproblems have been solved (P-\gls{admm}, Subsection~\ref{subsec:padmm}).
The resulting placement $\tbeta^{(k)}$ is then evaluated using the normalized objective
\begin{equation}
  \mathrm{obj}^{(k)} =
    \omega_{\text{risk}}\frac{\mathrm{Risk}^{\mathrm{LB}}}{\bar{R}}\,
      \bigl(\tbeta^{(k)}\bigr)
    + \omega_{\text{mig}}\frac{\mathrm{Mig}}{\bar{M}}\,
      \bigl(\tbeta^{(k)}\bigr),
  \label{eq:admm_incumbent}
\end{equation}
best-performing placement encountered thus far is retained as the incumbent solution.
The iterative procedure terminates when the incumbent objective exhibits negligible improvement, 
$|\mathrm{obj}^{(k)} - \mathrm{obj}^{(k-1)}| < \delta_{\mathrm{obj}}$, 
or when the maximum number of sweeps, $K_{\max}$, is reached.

Algorithm~\ref{alg:admm} summarizes the resulting penalized best-response decomposition. The scheduling flag determines whether coordination follows the S-\gls{admm} update rule of Subsection~\ref{subsec:sadmm} or the P-\gls{admm} update rule of Subsection~\ref{subsec:padmm}.

\begin{algorithm}[t]
\caption{S-\gls{admm} / P-\gls{admm}}
\label{alg:admm}
\begin{algorithmic}[1]
\Input
  \glspl{sfc} $\{(n,\calF_n,\mathcal{U}_n)\}_{n=1}^{N}$;\;
   $\mathcal{G}^{(t)}$;\; $\bfbeta^{(t-1)}$;\;
  $(\bar{R},\bar{M},\bm{\pi},w_{n,n'}^{f},\rho,K_{\max},
   \delta_{\mathrm{obj}},
   \omega_{\text{risk}},\omega_{\text{mig}},\delta_{\mathrm{sub}})$;\; schedule $\in\{\textsc{S},\textsc{P}\}$
\Output Placement $\bfbeta^{(t)}$
\State $\mathbf{d}^{\mathrm{SP}}\!\leftarrow\!\textsc{Dijkstra}(\mathcal{G}^{(t)})$;\;
       $\alpha_{n,f,i,s}\!\leftarrow\!0$;\; $\gamma_{n,f,i,s}\!\leftarrow\!0$;\;
       $\ell^{(0)}_{n,s}\!\leftarrow\!0$;\;
       $\mathrm{obj}^{\star}\!\leftarrow\!\infty$;\; $\mathrm{obj}^{-}\!\leftarrow\!\infty$
\For{$k \leftarrow 1$ \textbf{to} $K_{\max}$}
  \If{schedule $=\textsc{S}$}
    \State $\ell_s\!\leftarrow\!0\;\forall\,s$
  \Else
    \State $L^{(k-1)}_{s}\!\leftarrow\!\sum_{n'=1}^{N}\ell^{(k-1)}_{n',s}\;\forall\,s$
  \EndIf
  \For{$n \leftarrow 1$ \textbf{to} $N$}\Comment{\textsc{P}: in parallel}
    \State $\phi^{(k)}_{n,f,i,s}\!\leftarrow\!
           \frac{\omega_{\text{risk}}}{\bar{R}}
           \!\sum_{n'\neq n}\!w_{n,n'}^{f}\,\alpha_{n',f,i,s}
           \;+\;\rho\!\left(\tfrac{1}{2}-\gamma_{n,f,i,s}^{(k-1)}\right)$
           \hfill\eqref{eq:admm_penalty}
    \State $\Lambda^{(k)}_{n,s}\!\leftarrow\!
           \begin{cases}\ell_s & (\textsc{S})\\ L^{(k-1)}_s-\ell^{(k-1)}_{n,s} & (\textsc{P})\end{cases}$
           \hfill\eqref{eq:residual_reserve}
    \State $[\bfbeta_n,\bfgamma_n]\!\leftarrow\!
           \textsc{SolveMILP}\!\bigl(
             \mathrm{sfc}_n,\,\mathcal{G}^{(t)},\,\bfbeta^{(t-1)},
             \{\mathcal{C}^{\mathrm{cpu}}_s\!-\!\Lambda^{(k)}_{n,s}\}_s,
             \bm{\phi}^{(k)}_n,\,\mathbf{d}^{\mathrm{SP}},
             \delta_{\mathrm{sub}}\bigr)$
           \hfill\eqref{eq:admm_xstep}
    \If{schedule $=\textsc{S}$}\Comment{Gauss--Seidel: refresh in place}
      \State $\ell_s\mathrel{+}=\text{CPU}_{n,s}$;\;
             $\alpha_{n,f,i,s}\!\leftarrow\!\gamma_{n,f,i,s}$
    \EndIf
  \EndFor
  \If{schedule $=\textsc{P}$}\Comment{Jacobi: refresh after all solves}
    \State $\ell^{(k)}_{n,s}\!\leftarrow\!\text{CPU}_{n,s}$;\;
           $\alpha_{n,f,i,s}\!\leftarrow\!\gamma_{n,f,i,s}$
  \EndIf
  \State $\tbeta\!\leftarrow\!\textsc{Merge}(\bfbeta_1,\ldots,\bfbeta_N)$
  \If{schedule $=\textsc{P}$}
    \State $\tbeta\!\leftarrow\!\textsc{RepairColocation}(\tbeta,I)$\hfill Alg.~\ref{alg:repair}
  \EndIf
  \State $\mathrm{obj}(\tbeta)\!\leftarrow\!
         \omega_{\text{risk}}\tfrac{\mathrm{Risk}^{\mathrm{LB}}}{\bar{R}}(\tbeta)
         + \omega_{\text{mig}}\tfrac{\mathrm{Mig}}{\bar{M}}(\tbeta)$
         \hfill\eqref{eq:admm_incumbent}
  \If{$\mathrm{obj}(\tbeta)<\mathrm{obj}^{\star}$}
    \State $\mathrm{obj}^{\star}\!\leftarrow\!\mathrm{obj}(\tbeta)$;\;
           $\bfbeta^{(t)}\!\leftarrow\!\tbeta$
  \EndIf
  \If{$|\mathrm{obj}(\tbeta)-\mathrm{obj}^{-}|<\delta_{\mathrm{obj}}$}\; \textbf{break}
  \EndIf
  \State $\mathrm{obj}^{-}\!\leftarrow\!\mathrm{obj}(\tbeta)$
\EndFor
\State \Return $\bfbeta^{(t)}$
\end{algorithmic}
\end{algorithm}

\subsection{S-\gls{admm}: Sequential Gauss--Seidel Coordination}
\label{subsec:sadmm}

In S-\gls{admm}, the $N$ per-slice $\mathbf{x}$-steps are solved sequentially according to a fixed ordering,
$n = 1, 2, \ldots, N$. Unlike a parallel update scheme, the decisions of each slice are made immediately available to all subsequent slices within the same sweep. As a result, both resource utilization and co-location occupancy are progressively updated as the algorithm advances through the slice sequence. Specifically, after the optimization of slice $n$, the corresponding CPU allocations are committed and incorporated into the residual-capacity calculation. Consequently, constraint~\eqref{eq:c3rem} for a later slice accounts for the resources already reserved by previously optimized slices and is evaluated using the remaining satellite capacity $\mathcal{C}_s^{\mathrm{cpu}}
-
\sum_{\substack{n' \in \calN \\ n' < n}}
\mathrm{CPU}_{n',s}^{(k)}$.
At the same time, the occupancy estimate is updated in place. For every slice that has already been solved,
$\alpha_{n',f,i,s}
\leftarrow
\gamma_{n',f,i,s}^{(k)},
\qquad n' < n$,
such that the co-location penalty in~\eqref{eq:admm_penalty} reflects the most recent placement commitments. Therefore, when solving slice $n$, the linear penalty coefficients already capture the instances claimed by earlier slices during the current sweep. This Gauss--Seidel coordination mechanism enables each slice to react to the latest network state rather than to stale occupancy information. By exposing capacity consumption and VNF-placement decisions immediately, later slices are naturally guided away from congested satellites and heavily shared VNF instances. The resulting reduction in effective cross-slice coupling typically improves solution stability and accelerates convergence, while avoiding the synchronization overhead associated with parallel coordination schemes.
The S-\gls{admm} scheduling strategy implemented in Algorithm~\ref{alg:admm} formalizes this sequential update process.

\subsection{P-\gls{admm}: Parallel (Jacobi) Update with Collision Repair}
\label{subsec:padmm}

P-\gls{admm} replaces the sequential Gauss--Seidel schedule of S-\gls{admm} with a fully parallel Jacobi update. At iteration $k$, all $N$ per-slice {x}-steps are solved concurrently using only the state available from iteration $(k-1)$. Consequently, each slice optimizes its placement against a frozen view of both resource utilization and VNF occupancy. Specifically, slice $n$ evaluates the residual CPU capacity of satellite $s$ as $\mathcal{C}_s^{\mathrm{cpu}} - \Bigl( L_s^{(k-1)}
- \ell_{n,s}^{(k-1)} \Bigr)$, where $L_s^{(k-1)}
= \sum_{n'=1}^{N}
\ell_{n',s}^{(k-1)}$ denotes the aggregate CPU load committed on satellite $s$ during the previous sweep. Likewise, the occupancy estimate $\bm{\alpha}^{(k-1)}$ remains fixed throughout the iteration. As a result, every per-slice subproblem depends exclusively on iteration-$(k-1)$ information and can therefore be solved independently. This removes all within-sweep dependencies and enables the $N$ optimization tasks to be dispatched simultaneously across $N$ parallel workers.

The main drawback of this parallelism is the loss of the immediate coordination mechanism available in S-\gls{admm}. Because occupancy information is not updated during the sweep, multiple slices may independently select the same VNF instance $(f,i,s)$ while observing identical occupancy estimates. Such conflicts are not visible to the stale penalty term and may therefore persist in the merged placement. To restore inter-slice coordination, P-\gls{admm} appends a deterministic \emph{collision-repair} stage after the parallel solve and before objective evaluation. The overall procedure is summarized by the P-\gls{admm} schedule of Algorithm~\ref{alg:admm} and the repair routine of Algorithm~\ref{alg:repair}.

\begin{algorithm}[t]
\caption{\textsc{RepairColocation}}
\label{alg:repair}
\begin{algorithmic}[1]
\Input merged placement $\tbeta$;\; instances per satellite $I$
\Output repaired placement $\hbeta$
\For{each (function type $f$, satellite $s$) activated in $\tbeta$}
  \State $\mathcal{N}_{f,s}\!\leftarrow\!\{\,n : \text{slice } n \text{ places } f \text{ on } s\,\}$
  \For{rank $r$, slice $n$ \textbf{in} \textsc{enumerate}$\bigl(\textsc{sort}(\mathcal{N}_{f,s})\bigr)$}
    \State reassign every $f$-VNF of slice $n$ on $s$ to instance $i\!\leftarrow\! r \bmod I$
  \EndFor
\EndFor
\State \Return $\hbeta$
\end{algorithmic}
\end{algorithm}

The repair procedure operates independently for each function--satellite pair $(f,s)$. For every such group, it first identifies the set $\mathcal{N}_{f,s}$ of slices that place function $f$ on satellite $s$. Then, the slices are ordered deterministically and assigned instance identifiers according to their rank, with the slice of rank $r$ mapped to instance $r \bmod I$. This reassignment alters only the instance label and leaves the selected satellite unchanged. Consequently, satellite visibility constraints, routing decisions, E2E delay, and per-satellite resource consumption remain unaffected.

By redistributing slices across the $I$ available replicas, the repair eliminates unnecessary co-location whenever sufficient instances exist. In particular, if
$|\mathcal{N}_{f,s}| \le I$, every contending slice can be assigned a distinct instance, and all collisions are removed. When $|\mathcal{N}_{f,s}| > I$, the available instances are insufficient to provide complete isolation. In this oversubscribed regime, the rank-modulo assignment distributes slices as evenly as possible across the $I$ replicas, thereby minimizing concentration while leaving an unavoidable residual co-location. The corresponding residual risk is captured directly by the P-\gls{admm} objective and evaluated in Subsection~\ref{sec:results:scalability}.

The computational overhead of the repair stage is modest. Since each $(f,s)$ group is processed independently, the total complexity is $O\!\left(
\sum_{f,s} |\mathcal{N}_{f,s}|
\log |\mathcal{N}_{f,s}|
\right)$, arising from the sorting operation within each group. Moreover, the repair itself is parallel and can be executed concurrently across different function--satellite groups. After repair, the placement is evaluated using the objective in~\eqref{eq:admm_incumbent}, the incumbent solution is updated if improved, and the algorithm proceeds to the next sweep. The same objective-stagnation criterion and maximum-iteration limit used by S-\gls{admm} are adopted unchanged.
Finally, note that the deterministic ordering used by Algorithm~\ref{alg:repair} is selected solely for reproducibility. For a fixed function type and satellite, all $I$ instances are homogeneous replicas with identical visibility, capacity, and operating cost. Consequently, the sorted assignment introduces no systematic advantage for any slice and affects only the identity of the slices involved in the rare case of residual collisions when $|\mathcal{N}_{f,s}| > I$. Alternative tie-breaking strategies, such as randomized or load-balancing assignments, would merely redistribute the remaining collisions among different slice pairs without changing their overall number.

\subsection{Convergence Analysis}
\label{subsec:convergence}
The proposed decomposition is not a classical consensus-\gls{admm} algorithm because it does not maintain dual variables or perform a dual-ascent update. Instead, coordination is achieved through a combination of occupancy-based pricing and proximal regularization embedded in the per-slice subproblems. Consequently, the relevant convergence notion is not primal--dual convergence, but rather convergence to a stable best-response configuration of the penalized decomposition.

\begin{proposition}[Incumbent monotonicity and best-response fixed points]
\label{prop:convergence}
Let $\mathrm{obj}^{\star}_k$ be the incumbent objective after sweep~$k$ of Algorithm~\ref{alg:admm}. Then, the following properties hold:
\begin{enumerate}
  \item[(i)] The incumbent sequence $\{\mathrm{obj}^{\star}_k\}_{k\ge0}$ is non-increasing and bounded below by 0. Consequently, it converges, and the algorithm terminates after a finite number of sweeps, no later than $K_{\max}$.
  \item[(ii)] If the occupancy estimate remains unchanged between two consecutive sweeps, $\bm{\alpha}^{(k)} = \bm{\alpha}^{(k-1)}$, then each per-slice activation $\bfgamma_n^{(k)}$ is an optimal solution of its penalized subproblem~\eqref{eq:admm_xstep} under the fixed occupancy generated by the remaining slices. Hence, the collection of per-slice solutions constitutes a coordinate-wise best-response fixed point of the penalized decomposition.
  
\end{enumerate}
\end{proposition}

\begin{proof}
(i) The incumbent solution is updated only when a placement with a strictly lower objective value is found(Algorithm~\ref{alg:admm}). Hence, $\mathrm{obj}^{\star}_k \le \mathrm{obj}^{\star}_{k-1}$ and the incumbent sequence is monotone non-increasing. The normalized objective defined in~\eqref{eq:admm_incumbent} is a weighted sum of the normalized risk and migration costs, both of which are nonnegative by construction. Therefore,
$\mathrm{obj}^{\star}_k \ge 0, \forall k$. Since every monotone non-increasing sequence that is bounded below converges, $\mathrm{obj}^{\star}_k$ converges to a finite limit. The stopping rule based on objective stagnation together with the explicit iteration limit $K_{\max}$ guarantees finite termination.

(ii) Assume that
$\bm{\alpha}^{(k)}
=\bm{\alpha}^{(k-1)}$. In this case, the occupancy-dependent coefficients appearing in the penalized objective~\eqref{eq:admm_penalty} remain unchanged between sweeps. Since the proximal anchor is also fixed at the previous activation state, each slice solves the same penalized optimization problem as in the preceding sweep.
By construction, the {x}-step~\eqref{eq:admm_xstep} computes a solution that is optimal (up to the subproblem tolerance $\delta_{\mathrm{sub}}$) for slice $n$ under these fixed coefficients. Hence, given the occupancy induced by the other slices, no unilateral modification of $\bfgamma_n^{(k)}$ can further reduce slice $n$'s penalized objective. This property holds simultaneously for all slices, implying that the resulting placement is a coordinate-wise best-response equilibrium of the penalized decomposition.
For P-\gls{admm}, the collision-repair stage does not alter this argument. The occupancy estimate $\bm{\alpha}$ is updated directly from the per-slice activation decisions $\bfgamma_n^{(k)}$ produced by the parallel {x}-steps and not from the repaired placement. Consequently, the repair procedure is external to the iterative coordination mechanism and serves only to construct the placement evaluated by the incumbent objective. Statement~(ii) thus applies identically to both the S-\gls{admm} and P-\gls{admm} schedules.
Moreover, because the repair mapping of Algorithm~\ref{alg:repair} is deterministic, a fixed occupancy pattern generates the same merged placement at every subsequent sweep and thus the same repaired placement. Hence, once a best-response fixed point is reached, both the optimization variables and the final deployed placement remain unchanged across iterations.
\end{proof}

\subsection{Problem Size and Computational Cost}
\label{subsec:complexity}

The main motivation for the proposed decomposition is to remove the dominant cross-slice coupling that makes the monolithic formulation difficult to solve. In the original joint \gls{milp}, the largest combinatorial component arises from the co-location variables $y$, whose cardinality scales as
$O\!\left(N^{2}L_{\max}^{2}IS\right)$, since they are defined for pairs of slices, service-function positions, instances, and satellites. These variables, together with the auxiliary coupling variables $z$, introduce the quadratic dependence on the number of slices and significantly increase the branch-and-bound search space. In the proposed decomposition, both $y$ and $z$ are removed from the per-slice optimization problems. Instead, their effect is captured through the occupancy-based linear penalty $\phi^{(k)}$ applied to the activation variables $\gamma$ in~\eqref{eq:admm_penalty}. Consequently, the decision space of each per-slice {x}-step contains only $O(\bar{U}L_{\max}IS)$ assignment variables $\beta_n$, $O(|\calF_n|IS)$ activation variables $\gamma_n$, and $O(\bar{U}L_{\max})$ migration variables $\mu_n$. Also, the delay linearization of \eqref{C5} introduces continuous auxiliary variables $v^{s,s'}_{n,u,l}\in[0,1]$,
whose number scales as
$O\!\left( \bar{U}L_{\max}|\mathcal{E}^{(t)}| \right)$,
where $|\mathcal{E}^{(t)}|$ denotes the number of reachable satellite links during the orchestration epoch. Importantly, the size of each subproblem is independent of the total number of slices $N$ except through the externally supplied penalty coefficients.

Let $T_{\mathrm{sub}}$ be the runtime required to solve a single per-slice \gls{milp} to the target optimality gap $\delta_{\mathrm{sub}}$. Since each subproblem remains an integer program, its worst-case complexity is exponential in the number of binary variables. Hence, we treat $T_{\mathrm{sub}}$ as a solver-dependent primitive and evaluate it empirically in Subsection~\ref{sec:results:scalability}. Each sweep of Algorithm~\ref{alg:admm} performs $N$ such subproblem solves, followed by a merge operation whose complexity is $O\!\left(
N\bar{U}L_{\max}IS
\right)$. The merged placement is then evaluated using the normalized objective~\eqref{eq:admm_incumbent}. The dominant component of this evaluation is the pairwise co-location risk calculation, which requires
$O\!\left(
N^{2}|\calF|IS
\right)$ operations. For P-\gls{admm}, an additional collision-repair stage is executed. From the analysis of Algorithm~\ref{alg:repair}, its complexity is
$O\!\left(
\sum_{f,s}
|\mathcal{N}_{f,s}|
\log |\mathcal{N}_{f,s}|
\right)$, which is upper bounded by $O\!\left(
N|\calF|S\log N \right)$.
The overall computational work up to the maximum number of sweeps is therefore $O\!\left(
K_{\max} \left[
N\,T_{\mathrm{sub}} +
N^{2}|\calF|IS
\right]\right)$.

Although the decomposition does not eliminate the quadratic dependence on $N$ entirely, it removes this dependence from the optimization model itself. The $O(N^{2})$ interaction term survives only in the closed-form incumbent evaluation and no longer appears as binary decision variables within a \gls{milp}. In practice, this scoring cost is negligible compared with the cumulative subproblem solution time $N\,T_{\mathrm{sub}}$ across all tested scenarios.

The distinction between the two schedules lies primarily in their exploitable parallelism. In S-\gls{admm}, the Gauss--Seidel update order introduces a dependency between consecutive slices, forcing the $N$ subproblems to be solved sequentially. The resulting wall-clock complexity scales as $O\!\left(
K_{\max}N\,T_{\mathrm{sub}}
\right)$. In contrast, the Jacobi structure of P-\gls{admm} makes every subproblem depend only on iteration-$(k-1)$ information. Assuming the availability of $N$ workers, all per-slice solves can be executed concurrently, reducing the dominant wall-clock cost to
$O\!\left(K_{\max}T_{\mathrm{sub}}\right)$, excluding the merge, repair, and scoring phases. These latter stages remain at least linear in $N$ and thus limit the achievable speedup below the ideal $N$-fold factor predicted by perfect parallel scaling. The decomposition does not lower the worst-case complexity as each subproblem is still NP-hard. However, it replaces one large coupled MILP with $N$ small independent ones whose empirical solve times stay within the per-epoch budget.

\section{Simulation Setup and Compared Methods}
\label{sec:setup}

\subsection{Simulation Setup}
Using the MATLAB satellite communication toolbox (SatCom), we simulate a Walker-Delta \gls{leo} constellation \cite{Pachler2024} comprising $S = 80$ satellites distributed across $P = 8$ orbital planes at an altitude of $h_{\mathrm{orb}} = 550$\,km with an inclination of $53^{\circ}$.
Orbital propagation follows Keplerian dynamics with epoch snapshots generated every $\Delta t \in \{60,\,120,\,180,\,360\}$\,sec. Unless otherwise stated, $\Delta t = 120$\,sec is used as the reference epoch duration, a granularity at which the \gls{isl} topology is approximately stationary within a time window. Yet re-optimization remains frequent enough to track orbital motion~\cite{Wang2022}.
Each satellite maintains \glspl{isl} to its $k = 4$ nearest neighbors within a maximum range of $3{,}200$\,km~\cite{Pachler2026}.
The maximum concurrent flows per \gls{isl} link is $F_{\max} = 10$~\cite{Xu2024}.
The minimum elevation angle set to $\theta_{\min} = 10^{\circ}$, a conservative operational minimum consistent with \gls{3gpp} non-terrestrial-network guidance~\cite{3gpp38811}.

The security workload is configured to span the chain lengths and sensitivity levels representative of multi-tenant security slicing.
Each network slice $n$ carries a  \gls{sfc} consisting of $L_n \in \{3,4,5\}$ security \gls{vnf}s drawn from the function types $\mathcal{F} = \{\mathrm{\gls{fw}},\,\mathrm{\gls{ids}},\,\mathrm{ENC},\,\mathrm{\gls{tm}},\,\mathrm{\gls{siem}}\}$, a range consistent with practical \gls{sfc}s~\cite{Wang2025}.
Each function type $f$ has security sensitivity as follows: 
$R_{\mathrm{\gls{fw}}}=0.6$, 
$R_{\mathrm{\gls{ids}}}=0.8$, 
$R_{\mathrm{ENC}}=0.9$, 
$R_{\mathrm{\gls{tm}}}=0.4$, and 
$R_{\mathrm{\gls{siem}}}=0.6$. 
These values follow the ordering established in Definition~\ref{def:risk_weight}, assigning the highest sensitivity to functions that handle cryptographic material or attack-detection state (ENC, \gls{ids}) and the lowest to passive monitoring (\gls{tm}). Consequently, co-locating the most security-critical functions incurs the largest risk weight.
Migration disruption costs are $D_f^{\text{mig}} \in \{0.8,1.0,1.5,1.7,2.0\}$ for \gls{tm}, \gls{fw}, \gls{ids}, \gls{siem}, and ENC, respectively. These values are ordered to reflect the run-time state each function must transfer on relocation: stateless or lightly stateful functions (\gls{tm}, \gls{fw}) are the cheapest to migrate, whereas functions carrying detection or session state (\gls{ids}, \gls{siem}, ENC) incur progressively higher disruption.
Each slice has criticality $C_n$ drawn uniformly from $[1.0,\,3.0]$ and per-user E2E latency budget $\bar{T}_{n,u}$ drawn from $[100,\,150]$\,ms, spanning the Quality-of-service (QoS) targets of representative \gls{leo} verticals as set out in \gls{3gpp}/ITU slicing requirements~\cite{ETSI_MEC}. The criticality range spans the low-, medium-, and high-classification tiers of the data-sensitivity scale underlying $C_n$ in Definition~\ref{def:risk_weight}, such that slices contribute to co-location risk in proportion to the confidentiality of the traffic they carry. The isolation policy coefficient is $\Phi_{n,n'} \in [0,\,0.3]$, spanning from fully enforced inter-slice isolation ($\Phi_{n,n'}=0$) to a weak isolation policy, thereby exercising both the zero-trust and the partially-isolated operating regimes of the risk model. 
Security \gls{vnf} activation CPU costs are drawn from type-specific ranges $[0.3,\,2.5]$\,GFLOPS and per-user incremental costs from $[0.05,\,0.6]$\,GFLOPS, with satellite CPU capacities drawn from $[40,\,120]$\,GFLOPS.

To ensure that the reported trends reflect the decomposition rather than a single favourable draw, every operating point is averaged over repeated trials with confidence intervals.
Unless otherwise stated, experiments use $N = 10$ slices with $5$ users per slice, $I = 3$ candidate instances per function per satellite, $2$ independent random trials of $20$ evaluation epochs each (i.e, $40$ trial-epoch samples, over which means and $95\%$ confidence intervals are computed), and objective weights $\omega_{\text{risk}} = 0.75$, $\omega_{\text{mig}} = 0.25$. This setup is referred to as the \textit{base configuration} throughout this section.
Here, a trial denotes one random realization of the workload, such as SFCs, risk parameters, and per-\gls{vnf} CPU costs over the same constellation.

\subsection{Compared Methods}
\label{sec:setup:methods}

We compare the proposed methods against baselines that span the design space from solver-free heuristics to security-unaware
adaptive MILPs. The list of methods is:

\begin{itemize}
  \item \textbf{\gls{ru}:} the joint \gls{milp} with the security penalty disabled ($\omega_{\text{risk}} = 0$), resolved every epoch. It is the strongest non-secure baseline, whose gap to the proposed methods quantifies the security benefit of explicit co-location optimization.

  \item  \textbf{Greedy Heuristic (GDY):} a solver-free breadth-first search (BFS) with nearest-feasible-satellite assignment. 
  The entire service chain is reconstructed independently at each epoch without considering prior placements, migration costs, or co-location risk. GDY represents the lowest-complexity approach and provides a practical lower-bound benchmark for solution quality.

  \item  \textbf{Monolithic Security-Aware \gls{milp} (MYO):} the joint single-epoch security \gls{milp}~\eqref{eq:milp_obj} is solved centrally over all $N$ slices with a $\delta_{\text{MIP}} = 0.5\%$ gap. It directly optimizes the security-aware objective without decomposition. However, its computational cost increases rapidly with the number of slices due to the quadratic growth of the inter-slice coupling variables.

  \item \textbf{\textbf{S-ADMM/P-ADMM}:} the proposed penalized best-response decomposition with the sequential Gauss--Seidel schedule (S-\gls{admm}) and the parallel Jacobi schedule with collision repair (P-\gls{admm}), respectively. Unless otherwise stated, both variants use $\rho = 3.0$, $K_{\max} = 15$,
    $\delta_{\mathrm{obj}} = 10^{-4}$, a per-subproblem optimality gap of $\delta_{\mathrm{sub}} = 0.5\%$, and a $10$~sec time limit for each per-slice \gls{milp}.
\end{itemize}

\section{Results and Analysis}
\label{sec:results}

\subsection{Main Performance Comparison}
\label{sec:results:main}

\begin{figure*}[t]
  \centering
  \begin{subfigure}[t]{0.32\textwidth}
    \centering
    \includegraphics[width=\linewidth]{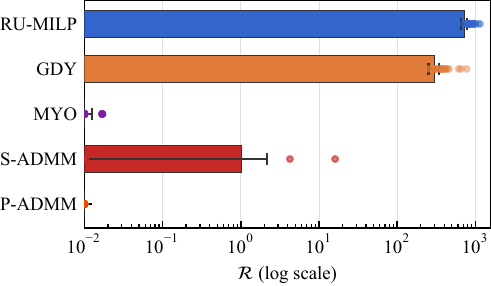}
    \caption{Average co-location risk $\mathcal{R}$}
    \label{fig:main_comparison_risk}
  \end{subfigure}\hfill
  \begin{subfigure}[t]{0.32\textwidth}
    \centering
    \includegraphics[width=\linewidth]{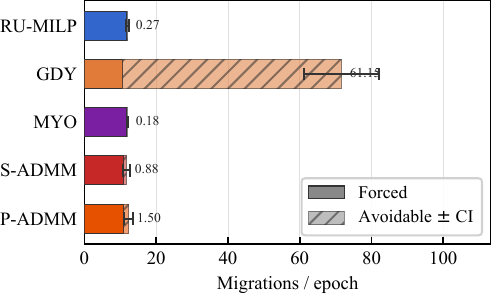}
    \caption{Average migrations per epoch}
    \label{fig:main_comparison_mig}
  \end{subfigure}\hfill
  \begin{subfigure}[t]{0.32\textwidth}
    \centering
    \includegraphics[width=\linewidth]{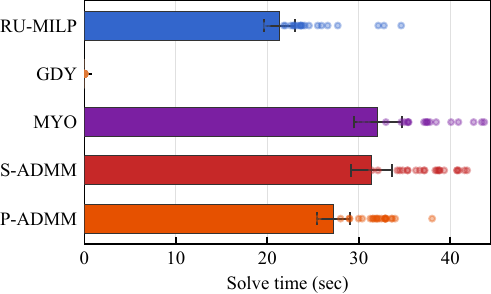}
    \caption{Average solve time}
    \label{fig:main_comparison_time}
  \end{subfigure}
  \caption{Average performance under the base configuration.}
  \label{fig:main_comparison}
\end{figure*}

Fig.~\ref{fig:main_comparison} presents the average results, in terms of co-location risk, number of migrations per epoch, and solve time (in sec), of all five methods under the base configuration when $N=10$. This is the largest slice count at which MYO remains solvable and thus the only setting admitting a direct five-method comparison across all performance dimensions.

\textit{Security co-location risk}: 
Accorinding to Fig.~\ref{fig:main_comparison_risk}, MYO and P-\gls{admm} eliminate co-location risk with $\mathcal{R} = 0.005$ for MYO and $\mathcal{R} = 0$ for P-\gls{admm}, while S-\gls{admm} reduces risk to $\mathcal{R} = 1.02$. These security-aware methods stay three orders of magnitude below \gls{ru} ($\mathcal{R} = 720.2$), confirming that ignoring the security term produces placements that regularly co-locate sensitive security \gls{vnf} pairs as the workload grows. Finally, GDY achieves $\mathcal{R} = 297.9$, which is lower than \gls{ru}'s result since the geographic spread of \glspl{sfc} provides sub-optimal incidental isolation.

\textit{Avoidable and forced VNF migrations}:
Based on Fig.~\ref{fig:main_comparison_mig}, where the avoidable portion of each bar is annotated with its numeric value so the small counts remain legible, MYO achieves an average of $0.18$ avoidable migrations per epoch (hatched bars), while S-\gls{admm} and P-\gls{admm} incur $0.88$ and $1.50$, respectively. The per-slice decomposition of the proposed methods relaxes global continuity coordination, resulting in a slight increase in controllable churn relative to MYO. P-\gls{admm}'s collision repair trades a small additional churn for tighter co-location separation. In contrast, \gls{ru} retains prior assignments almost entirely, achieving near-zero avoidable migrations ($0.27$ per epoch).  At the opposite, GDY produces $61.2$ avoidable migrations per epoch due to its epoch-by-epoch stateless reassignment. The forced migration rate is $\approx 10.5$--$11.8$ per epoch for all \gls{milp}-based methods, driven by topology-force handovers when a satellite drops below the elevation angle threshold.

\textit{Solve time}:
In Fig.~\ref{fig:main_comparison_time}, the decomposed methods return a feasible placement on all trial-epoch samples, except for MYO, which fails on $7$ of them in the base configuration.
Specifically, MYO requires an average of $32.1$,s per epoch across the $33$ epochs for which it successfully converges; however, it exceeds the imposed solver time limit of $300$,s in the remaining $7$ epochs. P-\gls{admm} achieves an average runtime of $27.2$,s per epoch, making it approximately $1.15\times$ faster than S-\gls{admm} while maintaining greater reliability than MYO. In contrast, GDY and RU-\gls{milp} exhibit significantly lower computational overhead. GDY completes execution almost instantaneously ($\approx 0.01$,s per epoch) due to its single-pass BFS-based reassignment strategy, whereas RU-\gls{milp} is the most computationally efficient \gls{milp} formulation because it eliminates the $O(N^2)$ co-location variables. Nevertheless, these computational gains come at the expense of substantially higher co-location risk, as illustrated in Fig.~\ref{fig:main_comparison_risk}. Consequently, the proposed decomposition-based approach offers the most computationally efficient security-aware solution, while remaining considerably more tractable than the joint \gls{milp} formulation at $N=10$.

\begin{figure}[t]
  \centering
  \includegraphics[width=0.85\columnwidth]{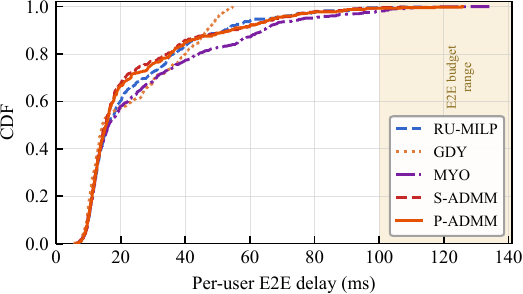}
  \caption{CDF of per-user \gls{e2e} delay.}
  \label{fig:delay_cdf}
\end{figure}

\textit{Delay distribution}: 
Fig.~\ref{fig:delay_cdf} plots the cumulative distribution functions (CDFs) of the E2E delay of the compared methods. 
Accordingly, we see that every user's E2E delay remains at or below $135$\,ms, which is within the defined delay thresholds $[100,150]$\,ms. Nevertheless, almost all methods (except MYO) ensure that 100\% of all users' E2E delay is below 100 ms. 

\textit{Epoch-by-epoch variability}: 
Fig.~\ref{fig:temporal} presents the per-epoch behaviours of risk and solve time over the $20$ evaluation epochs. 
In the risk panel (Fig.~\ref{fig:temporal_risk}), each plotted value is the average over the two random trials at that epoch index; the solve-time results are instead shown in a separate panel per trial (Figs.~\ref{fig:temporal_time1} and~\ref{fig:temporal_time2}), so that the epochs where MYO times out appear as breaks rather than being masked by trial averaging. Fig.~\ref{fig:temporal_risk} shows that MYO, S-\gls{admm}, and P-\gls{admm} consistently maintain low risk across all epochs, while GDY and \gls{ru} exhibit higher and more variable risk, due to the absence of security-aware objectives.
\begin{figure*}[t]
  \centering
  \begin{subfigure}[t]{0.32\textwidth}
    \centering
    \includegraphics[width=\linewidth]{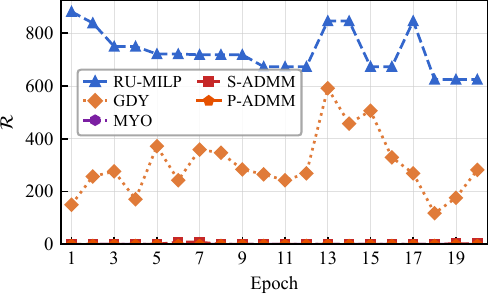}
    \caption{Per-epoch security risk ($\mathcal{R}$)}
    \label{fig:temporal_risk}
  \end{subfigure}\hfill
  \begin{subfigure}[t]{0.32\textwidth}
    \centering
    \includegraphics[width=\linewidth]{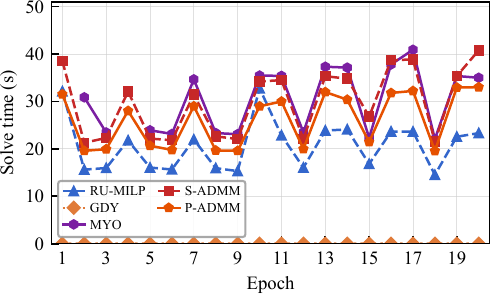}
    \caption{Per-epoch solve time (trial 1)}
    \label{fig:temporal_time1}
  \end{subfigure}\hfill
  \begin{subfigure}[t]{0.32\textwidth}
    \centering
    \includegraphics[width=\linewidth]{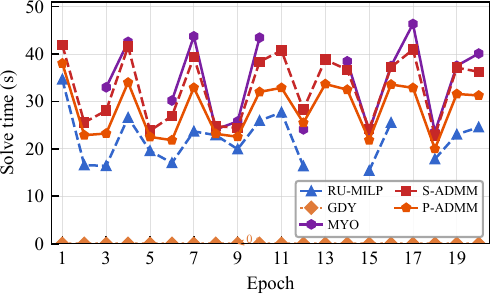}
    \caption{Per-epoch solve time (trial 2)}
    \label{fig:temporal_time2}
  \end{subfigure}
  \caption{Per-epoch temporal dynamics: (a) security risk averaged over the two trials, and (b, c) solve time for each trial.}
  \label{fig:temporal}
\end{figure*}

In Figs.~\ref{fig:temporal_time1} and~\ref{fig:temporal_time2},
MYO returns a feasible solution on only $33$ of $40$ trial-epoch samples and times out on the remaining $7$ (two in trial~1 and five in trial~2). Since each trial is drawn in its own panel, these timed-out epochs appear as breaks in MYO's solid line, where no feasible solve time exists.
For the other MILP methods, the solve times are in the range $[15,40]$ sec per epoch, with RU-MILP being the fastest solution. Nevertheless, the proposed ADMM methods achieve faster times than MYO, while ensuring that all epochs are resolved.

\subsection{Scalability Analysis}
\label{sec:results:scalability}

In Fig.~\ref{fig:scalability} , we present the averaged results, in terms of co-location risk, migrations per epoch, and solve time, for the considered methods, as functions of the number of slices $N \in \{5,7,10,12,15,20,25\}$.  
\begin{figure*}[t]
  \centering
  \begin{subfigure}[t]{0.32\textwidth}
    \centering
    \includegraphics[width=\linewidth]{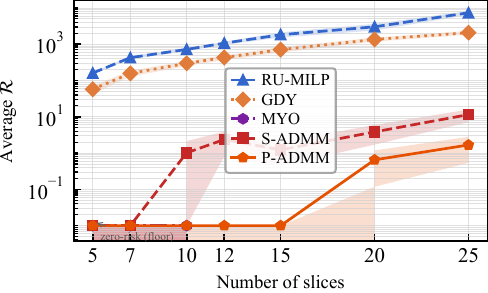}
    \caption{Log scaled average security risk ($\mathcal{R}$)}
    \label{fig:scalability_risk}
  \end{subfigure}\hfill
  \begin{subfigure}[t]{0.32\textwidth}
    \centering
    \includegraphics[width=\linewidth]{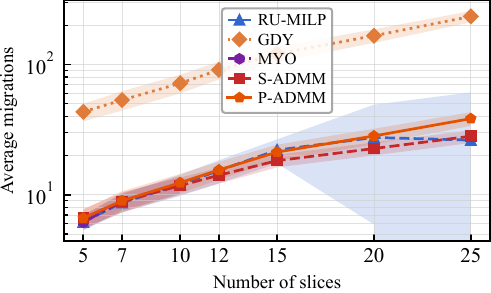}
    \caption{Log scaled average migrations per epoch}
    \label{fig:scalability_mig}
  \end{subfigure}\hfill
  \begin{subfigure}[t]{0.32\textwidth}
    \centering
    \includegraphics[width=\linewidth]{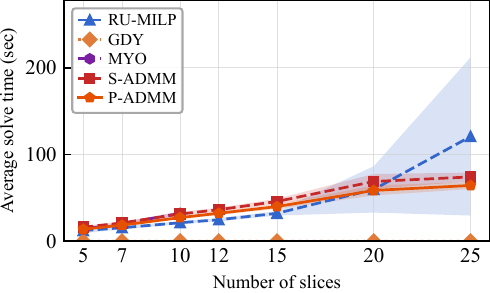}
    \caption{Average solve time per epoch (s)}
    \label{fig:scalability_time}
  \end{subfigure}
  \caption{Scalability versus number of slices.}
  \label{fig:scalability}
\end{figure*}

\textit{Security risk scaling}: According to Fig.~\ref{fig:scalability_risk}, MYO reaches near-zero risk wherever feasible ($N \leq 10$), but collapses (no solution) for $N>10$.
In contrast, P-\gls{admm} holds risk at exactly zero for $N \leq 15$, with only a small residual ($0.6$ at $N=20$ and $1.68$ at $N = 25$), as its deterministic repair separates every co-located pair that the $I$ instances can handle. S-\gls{admm}'s residual risk grows steadily ($\mathcal{R} \in [1.02, 11.50]$ for $N \in [10,25]$) as the Gauss--Seidel updates fail under heavier coupling. 
Both \gls{admm} methods remain orders of magnitude below the non-secure baselines, whose risk increases linearly (e.g., GDY reaches $2084.6$ and \gls{ru} $7333$ for $N = 25$). \gls{ru}'s feasibility also collapses for high $N$ values, e.g., only $4$ epochs were solved for $N = 25$.
Unlike MYO, \gls{ru} omits the $O(N^2)$ co-location variables. Hence, its infeasibility stems from the delay-constrained assignment, rather than co-location. Since unsolved epochs are the most contended, this feasible-only average understates \gls{ru}'s true risk. In all panels of Fig.~\ref{fig:scalability}, the solid line is the mean over the trial-epoch samples and the shaded band is the $95\%$ confidence interval of that mean.

\textit{Migration scaling}:
Based on Fig. \ref{fig:scalability_mig}, 
the total number of migrations per epoch (the quantity plotted) scales linearly with $N$ for MILP-based methods, consistent with a fixed topology-driven handover rate per epoch. S-\gls{admm} keeps total migrations low and close to MYO, whereas P-\gls{admm} consistently incurs more (e.g.\ $38.4$ vs.\ $28.1$ total migrations per epoch at $N = 25$, and $21.2$ vs.\ $18.2$ at $N = 15$). The gap is driven by P-\gls{admm}'s collision repair, which relocates additional instances to lower co-location risk, trading migration churn for tighter security; this is the risk-versus-migration trade-off that distinguishes the two schedules. As in Fig.~\ref{fig:scalability}, the shaded bands denote the $95\%$ confidence interval of the per-epoch mean.

\textit{Solve time scaling}: According to Fig. \ref{fig:scalability_time}, MYO's solve cost grows steeply with $N$ and exits the comparison beyond $N = 10$ (cf.\ the feasibility limit noted above).
P-ADMM and S-ADMM stay feasible on all $40$ epochs at every scale and well within the $120$\,sec epoch budget, with mean solve times growing from
$16$ to $74$\,s for S-\gls{admm} and $14$ to $64$\,s for P-\gls{admm} for $N$ from $5$ to $25$. The security-unaware baselines bound the two extremes: GDY remains near-instantaneous at every $N$, while RU-MILP is the cheapest \gls{milp} but its solve time is reported over a shrinking feasible subset as its feasibility degrades at large $N$ (e.g.\ $4/40$ epochs at $N = 25$).

\section{Conclusion}
\label{sec:conclusion}

This paper addresses security- and migration-aware \gls{sfc} placement in \gls{leo} satellite constellations, where an epoch-by-epoch rotating topology couples confidentiality, service continuity, and computational tractability in ways absent from terrestrial \gls{nfv} placement. 
The primary contribution is a unified placement framework built on two modelling advances: 
i)  a co-location security risk model grounded in established ISO/NIST risk-management principles, equipped with analytic bounds that render cross-slice exposure a quantifiable optimization objective and 
ii) an avoidable-versus-forced migration formulation that penalizes only controllable churn against the unavoidable handover baseline imposed by orbital motion. These are cast as the first multi-slice security-aware single-epoch \gls{milp} for \gls{leo} \gls{sfc} placement, subject
to per-satellite CPU, \gls{isl}, and per-user \gls{e2e} delay constraints. 
To render this formulation deployable at practical slice counts, we employed an \gls{admm}-style penalized per-slice best-response decomposition as the enabling solution technique. It transfers the quadratically growing cross-slice co-location coupling into a linear per-slice penalty. This decomposition replaces the intractable joint program with $N$ independent subproblems coordinated through two complementary schedules: i) a sequential Gauss--Seidel refresh (S-\gls{admm}) and ii) a
parallel Jacobi schedule with deterministic collision repair (P-\gls{admm}).

Evaluation over an $80$-satellite Walker--Delta constellation across a range of slice counts yields three principal insights. 
First, the decisive practical benefit is feasibility, as the framework returns a placement within the per-epoch time budget at every tested scale and sustains full \gls{e2e} delay compliance. In contrast, the monolithic security-aware \gls{milp} becomes intractable beyond a modest slice count. Therefore, the proposed framework operates precisely in the regime where the joint solver cannot. 
Second, embedding the risk model directly in the objective is what secures the placement. Both schedules suppress co-location risk by orders of magnitude relative to security-unaware optimization, with the parallel schedule holding risk at or near zero as the constellation scales.
Third, the two schedules expose a controllable risk-versus-migration trade-off at a near-identical overall objective. The parallel schedule attains lower co-location risk for modestly more migration, while its independent schedule yields a secondary wall-clock speedup. Collectively, these results show that the proposed framework matches or exceeds the security of the joint optimum while scaling to operating regimes that the monolithic formulation cannot reach.

Some limitations delineate the scope of these findings and motivate future work. The decomposition guarantees a monotone, bounded incumbent and a coordinate-wise best-response equilibrium at its fixed point
(Proposition~\ref{prop:convergence}) rather than convergence to the joint optimum. Additionally, the realized parallel speedup is bounded by a thread-based implementation in which only the per-slice solver calls execute concurrently. Building on this foundation, future work includes a predictive multi-epoch extension that pre-positions \glspl{vnf} across orbital handoffs, a process-parallel or distributed implementation that realizes the decomposition's full speedup, and validation across heterogeneous multi-orbit configurations. By making cross-slice co-location risk a quantifiable and tractable objective, this framework provides a tractable basis for secure, continuity-aware service orchestration as on-orbit edge computing scales toward 6G deployment.

\bibliographystyle{IEEEtran}
\bibliography{references}

@techreport{etsi_nfv,
  author      = {{ETSI}},
  title       = {{GS NFV-SEC 024 - V1.1.1}:  Network Functions Virtualisation ({NFV}); Security; Security Management},
  institution = {European Telecommunications Standards Institute},
  year        = {2026}
}

@INPROCEEDINGS{cohen2015near,
  author={Cohen, Rami and Lewin-Eytan, Liane and Naor, Joseph Seffi and Raz, Danny},
  booktitle={IEEE Conference on Computer Communications (INFOCOM)}, 
  title={Near optimal placement of virtual network functions}, 
  year={2015},
  volume={},
  number={},
  pages={1346-1354}
  }

@article{eramo2017approach,
  author  = {Eramo, Vincenzo and Amici, Matteo and Germoni, Angelo},
  title   = {An Approach for the Dimensioning of Computation Resources in an {NFV} Infrastructure Hosting Virtual-Router Functions},
  journal = {IEEE Transactions on Network and Service Management},
  volume  = {14},
  number  = {3},
  pages   = {539--552},
  year    = {2017},
  publisher = {IEEE}
}

@article{rankothge2017optimizing,
  author  = {Rankothge, Windhya and Le, Franck and Russo, Alessandra and Lobo, Jorge},
  title   = {Optimizing Resource Allocation for Virtualized Network Functions in a Cloud Center Using Genetic Algorithms},
  journal = {IEEE Transactions on Network and Service Management},
  volume  = {14},
  number  = {2},
  pages   = {343--356},
  year    = {2017},
  publisher = {IEEE}
}

@article{biallach2024vnf,
  author  = {Biallach, Aymane and Nace, Dritan and Tomaszewski, Artur and Bouhtou, Mustapha and Kumbria, Kingsley},
  title   = {Virtual Network Function Reconfiguration in {5G} Networks: An Optimization Perspective},
  journal = {Networks},
  volume  = {84},
  number  = {1},
  pages   = {3--22},
  month   = feb,
  year    = {2024},
  publisher = {Wiley}
}

@article{wang2020sfc,
  author  = {Wang, Guangchao and Zhou, Sheng and Zhang, Shan and Niu, Zhisheng and Shen, Xuemin},
  title   = {{SFC}-Based Service Provisioning for Reconfigurable Space-Air-Ground Integrated Networks},
  journal = {IEEE Journal on Selected Areas in Communications},
  volume  = {38},
  number  = {7},
  pages   = {1478--1489},
  month   = jul,
  year    = {2020},
  publisher = {IEEE},
  doi     = {10.1109/JSAC.2020.2986851}
}

@article{huang2020service,
  author  = {Huang, Yan and Liu, Jun and Shi, Yi and Taleb, Tarik and Kato, Nei},
  title   = {Service Function Placement and Chaining in Satellite Edge Computing},
  journal = {IEEE Transactions on Network and Service Management},
  volume  = {17},
  number  = {4},
  pages   = {2462--2475},
  year    = {2020},
  publisher = {IEEE}
}

@article{he2024iotj,
  author  = {He, Yejun and Shen, Shuyu and Shi, Wenchao and Guo, Song},
  title   = {Load-Aware Network Resource Orchestration in {LEO} Satellite Network: A {GAT}-Based Approach},
  journal = {IEEE Internet of Things Journal},
  volume  = {11},
  number  = {9},
  pages   = {15969--15984},
  year    = {2024},
  publisher = {IEEE}
}

@article{minardi2025jsac,
  author  = {Minardi, Mario and Vu, Thang X. and Lei, Lei and Chatzinotas, Symeon and Ottersten, Bj{\"o}rn},
  title   = {{SAST-VNE}: A Flexible Framework for Network Slicing in {6G} Integrated Satellite-Terrestrial Networks},
  journal = {IEEE Open Journal of the Communications Society},
  volume  = {5},
  pages   = {5648--5663},
  year    = {2024},
  publisher = {IEEE}
}

@article{boyd2011distributed,
  author  = {Boyd, Stephen and Parikh, Neal and Chu, Eric and Peleato, Borja and Eckstein, Jonathan},
  title   = {Distributed Optimization and Statistical Learning via the Alternating Direction Method of Multipliers},
  journal = {Foundations and Trends in Machine Learning},
  volume  = {3},
  number  = {1},
  pages   = {1--122},
  year    = {2011},
  publisher = {Now Publishers},
  doi     = {10.1561/2200000016}
}

@article{yu2021distributed,
  author  = {Yu, Rui and Ding, Zhiguo and Zhong, Sheng and Wan, Xin and Yang, Shu and Zhang, Yan},
  title   = {{ADMM}-Based Distributed Algorithm for Energy-Efficient {VNF} Placement in Mobile Edge Networks},
  journal = {IEEE Transactions on Green Communications and Networking},
  volume  = {5},
  number  = {3},
  pages   = {1246--1259},
  year    = {2021},
  publisher = {IEEE}
}

@inproceedings{ristenpart2009hey,
  author    = {Ristenpart, Thomas and Tromer, Eran and Shacham, Hovav and Savage, Stefan},
  title     = {Hey, You, Get Off of My Cloud: Exploring Information Leakage in Third-Party Compute Clouds},
  booktitle = {Proceeding ACM Conference Computing  Communication Security (CCS)},
  pages     = {199--212},
  year      = {2009},
  doi       = {10.1145/1653662.1653687}
}

@inproceedings{liu2015last,
  author    = {Liu, Fangfei and Yarom, Yuval and Ge, Qian and Heiser, Gernot and Lee, Ruby B.},
  title     = {Last-Level Cache Side-Channel Attacks are Practical},
  booktitle = { IEEE Symposium on Security and Privacy},
  pages     = {605--622},
  year      = {2015},
  doi       = {10.1109/SP.2015.43}
}

@INPROCEEDINGS{ali2020security,
  author={Ali, Abeer and Anagnostopoulos, Christos and Pezaros, Dimitrios P.},
  booktitle={IEEE Symposium on Computers and Communications (ISCC)}, 
  title={In-Network Placement of Security {VNFs} in Multi-Tenant Data Centers}, 
  year={2020},
  volume={},
  number={},
  pages={1-6},
  doi={10.1109/ISCC50000.2020.9219711}}

@article{zhang2021security,
  author  = {Zhang, Zhuo and Yu, Le and Wu, Wei and Maharjan, Sabita},
  title   = {Security-Aware Virtual Network Function Placement in Edge-Cloud Networks},
  journal = {IEEE Transactions on Information Forensics and Security},
  volume  = {16},
  pages   = {3324--3337},
  year    = {2021},
  publisher = {IEEE}
}

@techreport{iso31000,
  author      = {{ISO}},
  title       = {{ISO 31000:2018} Risk Management -- Guidelines},
  institution = {International Organization for Standardization},
  year        = {2018}
}

@techreport{nist800,
  author      = {{NIST}},
  title       = {{SP 800-53 Rev. 5}: Security and Privacy Controls for Information Systems and Organizations},
  institution = {National Institute of Standards and Technology},
  year        = {2020},
  doi         = {10.6028/NIST.SP.800-53r5}
}

@article{mccormick1976computability,
author = {Mccormick, Garth P.},
title = {Computability of global solutions to factorable nonconvex programs: Part {I} -- Convex underestimating problems},
year = {1976},
issue_date = {December  1976},
publisher = {Springer-Verlag},
address = {Berlin, Heidelberg},
volume = {10},
number = {1},
issn = {0025-5610},
journal = {Math. Program.},
pages = {147–175},
numpages = {29}
}

@ARTICLE{ETSI_MEC,
  author={ETSI},
  title={Mobile Edge Computing ({MEC}); Framework and Reference Architecture},
  journal={ETSI GS MEC 003 V2.1.1},
  year={2019},
  volume={},
  number={},
  pages={1-32},
  doi={10.3403/30322645U}
}

@INPROCEEDINGS{Jia2020,
  author={Jia, Ziye and Sheng, Min and Li, Jiandong and Zhu, Yan and Bai, Weigang and Han, Zhu},
  booktitle={IEEE International Conference on Communications (ICC)}, 
  title={Virtual Network Functions Orchestration in Software Defined {LEO} Small Satellite Networks}, 
  year={2020},
  volume={},
  number={},
  pages={1-6},
  doi={10.1109/ICC40277.2020.9148906}}

@ARTICLE{Jia2021,
  author={Jia, Ziye and Sheng, Min and Li, Jiandong and Zhou, Di and Han, Zhu},
  journal={IEEE Transactions on Wireless Communications}, 
  title={{VNF}-Based Service Provision in Software Defined {LEO} Satellite Networks}, 
  year={2021},
  volume={20},
  number={9},
  pages={6139-6153},
  doi={10.1109/TWC.2021.3072155}}

@INPROCEEDINGS{Yan2025,
  author={Yan, Weixin and Zhang, Zhilong and Liu, Danpu and Luo, Tao},
  booktitle={IEEE Wireless Communications and Networking Conference (WCNC)}, 
  title={Joint Optimization of {VNF} Reusing and Routing in Satellite Networks}, 
  year={2025},
  volume={},
  number={},
  pages={1-6},
  doi={10.1109/WCNC61545.2025.10978727}}

@ARTICLE{XGao2021,
  author={Gao, Xiangqiang and Liu, Rongke and Kaushik, Aryan},
  journal={IEEE Transactions on Network and Service Management}, 
  title={Service Chaining Placement Based on Satellite Mission Planning in Ground Station Networks}, 
  year={2021},
  volume={18},
  number={3},
  pages={3049-3063},
  doi={10.1109/TNSM.2020.3045432}}

@ARTICLE{XGao2022,
  author={Gao, Xiangqiang and Liu, Rongke and Kaushik, Aryan},
  journal={IEEE Transactions on Network and Service Management}, 
  title={Virtual Network Function Placement in Satellite Edge Computing With a Potential Game Approach}, 
  year={2022},
  volume={19},
  number={2},
  pages={1243-1259}}

@ARTICLE{Geng2024,
  author={Yuhui, Geng and Niwei, Wang and Xi, Chen and Xiaofan, Xu and Changsheng, Zhou and Junyi, Yang and Zhenyu, Xiao and Xianbin, Cao},
  journal={China Communications}, 
  title={Service Function Chain Migration in {LEO} Satellite Networks}, 
  year={2024},
  volume={21},
  number={3},
  pages={247-259}}

@ARTICLE{Xia2024,
  author={Xia, Qiufen and Wang, Guijie and Xu, Zichuan and Liang, Weifa and Xu, Zhou},
  journal={IEEE Transactions on Mobile Computing}, 
  title={Efficient Algorithms for Service Chaining in {NFV}-Enabled Satellite Edge Networks}, 
  year={2024},
  volume={23},
  number={5},
  pages={5677-5694}}

@ARTICLE{He2024,
  author={He, Jingchao and Cheng, Nan and Yin, Zhisheng and Zhou, Haibo and Zhou, Conghao and Aldubaikhy, Khalid and Alqasir, Abdullah and Shen, Xuemin},
  journal={IEEE Internet of Things Journal}, 
  title={Load-Aware Network Resource Orchestration in {LEO} Satellite Network: A {GAT}-Based Approach}, 
  year={2024},
  volume={11},
  number={9},
  pages={15969-15984}}

@ARTICLE{Doan2025,
  author={Doan, Khai and Avgeris, Marios and Leivadeas, Aris and Lambadaris, Ioannis and Shin, Wonjae},
  journal={IEEE Transactions on Vehicular Technology}, 
  title={Cooperative Learning-Based Framework for {VNF} Caching and Placement Optimization Over Low {Earth} Orbit Satellite Networks}, 
  year={2025},
  volume={74},
  number={3},
  pages={4758-4773}}

@ARTICLE{Liu2026,
  author={Liu, Yuru and Zhang, Xin and Zhang, Li and Deng, Lei and Gao, Fang and Zhou, Haibo},
  journal={IEEE Transactions on Vehicular Technology  (Early Access)}, 
  title={Robust Virtual Network Embedding for Ultra-Dense {LEO} Satellite Networks: An Adaptive Deep Reinforcement Learning Approach}, 
  year={2026},
  number={},
  pages={1-14}}

@ARTICLE{Ahsan2026,
  author={Ahsan, Muhammad and Vu, Thang X. and Chatzinotas, Symeon},
  journal={IEEE Transactions on Vehicular Technology}, 
  title={Flexible Resource Allocation and Path Reconfiguration Strategies for e{MBB} and m{MTC} Services in a {LEO} Satellite Topology}, 
  year={2026},
  volume={75},
  number={1},
  pages={642-653}}

@ARTICLE{Chen2026,
author={Chen, Xiao and Hu, Zenghao and Wang, Boyu and Zhu, Chao and Zhao, Zhenyu and Li, Xin and Wang, Fangxin},
journal={IEEE Transactions on Mobile Computing},
title={ {Sat-SFC}: Service Function Chain Placement in Battery Supply Varying Satellite Networks },
year={2026},
volume={PrePrints},
number={01},
ISSN={1558-0660},
pages={1-18},
doi={10.1109/TMC.2026.3692605},
publisher={IEEE Computer Society},
address={Los Alamitos, CA, USA}}

@ARTICLE{Liu2026a,
  author={Liu, Tianhao and Li, Xin and Zhang, Yongjun and Zhao, Chenyu and Yan, Xuhao and Huang, Shanguo},
  journal={IEEE Transactions on Green Communications and Networking}, 
  title={Cost-Efficient Service Function Chain Deployment in Satellite Optical Networks}, 
  year={2026},
  volume={10},
  number={},
  pages={2447-2462},
  doi={10.1109/TGCN.2026.3666650}}

@ARTICLE{Mahyoub2026,
  author={Mahyoub, Mohammed and Jaafar, Wael and Muhaidat, Sami and Yanikomeroglu, Halim},
  journal={IEEE Transactions on Network and Service Management}, 
  title={{STARS}: Stability-Aware {SFC} Orchestration and Associations in {LEO} Satellite Networks}, 
  year={2026},
  volume={23},
  number={},
  pages={3326-3340},
  doi={10.1109/TNSM.2026.3674391}}

@ARTICLE{Jia2025,
  author={Jia, Ziye and Cao, Yilu and He, Lijun and Wu, Qihui and Zhu, Qiuming and Niyato, Dusit and Han, Zhu},
  journal={IEEE Transactions on Vehicular Technology}, 
  title={Service Function Chain Dynamic Scheduling in Space-Air-Ground Integrated Networks}, 
  year={2025},
  volume={74},
  number={7},
  pages={11235-11248}}

@ARTICLE{Qin2023,
  author={Qin, Xiaohan and Ma, Ting and Tang, Zhixuan and Zhang, Xin and Zhou, Haibo and Zhao, Lian},
  journal={IEEE Transactions on Wireless Communications}, 
  title={Service-Aware Resource Orchestration in Ultra-Dense {LEO} Satellite-Terrestrial Integrated 6G: A Service Function Chain Approach}, 
  year={2023},
  volume={22},
  number={9},
  pages={6003-6017}}

@ARTICLE{Zheng2025,
  author={Zheng, Gao and Wang, Ning and Qian, Peng and Griffin, David and Tafazolli, Regius Rahim},
  journal={IEEE Journal on Selected Areas in Communications}, 
  title={{SDN}-Based Service Function Chaining in Integrated Terrestrial and {LEO} Satellite-Based Space Internet}, 
  year={2025},
  volume={43},
  number={2},
  pages={537-550}}

@ARTICLE{Liu2026b,
  author={Liu, Yepeng and Zhang, Ran and Liu, Jiang and Sun, Ninghan and Zhang, Xinyuan},
  journal={IEEE Transactions on Mobile Computing}, 
  title={{VMR-STAG} Based Online {SFC} Orchestration in Space-Terrestrial Integrated Networks}, 
  year={2026},
  volume={25},
  number={3},
  pages={3936-3952}}

@ARTICLE{Xu2024,
  author={Xu, Meilin and Jia, Min and Guo, Qing and de Cola, Tomaso},
  journal={IEEE Transactions on Vehicular Technology}, 
  title={Delay-Sensitive and Resource-Efficient {VNF} Deployment in Satellite-Terrestrial Networks}, 
  year={2024},
  volume={73},
  number={10},
  pages={15467-15482}}

@ARTICLE{Petrosino2023,
  author={Petrosino, Antonio and Piro, Giuseppe and Grieco, Luigi Alfredo and Boggia, Gennaro},
  journal={IEEE Transactions on Network and Service Management}, 
  title={On the Optimal Deployment of Virtual Network Functions in Non-Terrestrial Segments}, 
  year={2023},
  volume={20},
  number={4},
  pages={4831-4845}}

@ARTICLE{Wang2026,
  author={Wang, Peng and Sourav, Suman and Chen, Binbin and Li, Hongyan},
  journal={IEEE Transactions on Mobile Computing}, 
  title={An {SFC}-Constrained Max-Flow Solver for Satellite Networks Using Flexible Function-Time Expanded Graph}, 
  year={2026},
  volume={25},
  number={3},
  pages={4103-4120}}

@article{JGao2022,
title = {Resource consumption and security-aware multi-tenant service function chain deployment based on hypergraph matching},
journal = {Computer Networks},
volume = {216},
pages = {109298},
year = {2022},
issn = {1389-1286},
doi = {https://doi.org/10.1016/j.comnet.2022.109298},
author = {Jing Gao and Lei Feng and Peng Yu and Fanqin Zhou and Zihao Wu and Xuesong Qiu and Jingchun Li and Yifei Zhu}
}

@ARTICLE{Wang2024,
  author={Wang, Ben and Li, Jun and Cao, Shaohua and Guler, Evrim and Zheng, Danyang},
  journal={IEEE Access}, 
  title={Security-Aware Service Function Chaining and Embedding With Asymmetric Dedicated Protection}, 
  year={2024},
  volume={12},
  number={},
  pages={53944-53957}}

@ARTICLE{ZhangVNE,
  author={Zhang, Peiying and Wang, Chao and Jiang, Chunxiao and Benslimane, Abderrahim},
  journal={IEEE Transactions on Network Science and Engineering}, 
  title={Security-Aware Virtual Network Embedding Algorithm Based on Reinforcement Learning}, 
  year={2021},
  volume={8},
  number={2},
  pages={1095-1105}}

@ARTICLE{Wang2021,
  author={Wang, Xiangfeng and Yan, Junchi and Jin, Bo and Li, Wenhao},
  journal={IEEE Transactions on Cybernetics}, 
  title={Distributed and Parallel {ADMM} for Structured Nonconvex Optimization Problem}, 
  year={2021},
  volume={51},
  number={9},
  pages={4540-4552}}

@ARTICLE{Asheralieva2024,
  author={Asheralieva, Alia and Niyato, Dusit and Miyanaga, Yoshikazu},
  journal={IEEE Transactions on Mobile Computing}, 
  title={Efficient Dynamic Distributed Resource Slicing in {6G} Multi-Access Edge Computing Networks With Online {ADMM} and Message Passing Graph Neural Networks}, 
  year={2024},
  volume={23},
  number={4},
  pages={2614-2638}}

@article{Dubba2024,
  author    = {Sudha Dubba and Balaprakasa Rao Killi},
  title     = {Security-Aware Cost Optimized Dynamic Service Function Chain Scheduling},
  journal   = {Journal of Network and Systems Management},
  year      = {2024},
  volume     = {33},
  number     = {1},
  pages      = {4},
  doi        = {10.1007/s10922-024-09880-2},
  issn       = {1573-7705}
}

@ARTICLE{mahyoub2026_visibility,
  author={Mahyoub, Mohammed and Yanikomeroglu, Halim and Karabulut Kurt, Gunes and Martel, Stephane},
  journal={IEEE Transactions on Network and Service Management}, 
  title={Visibility-Aware User Association and Resource Allocation in Multi-Slice LEO Satellite Networks}, 
  year={2026},
  volume={23},
  number={},
  pages={1596-1614}}

@article{mahyoub_slicing_2025,
author = {Mahyoub, Mohammed and AbdulGhaffar, AbdulAziz and Alalade, Emmanuel and Matrawy, Ashraf},
title = {A Security-Aware Network Function Sharing Model for {5G} Slicing},
journal = {SECURITY AND PRIVACY},
volume = {8},
number = {3},
pages = {e70039},
year = {2025}
}

@INPROCEEDINGS{AbdulGhaffar2024,
  author={AbdulGhaffar, AbdulAziz and Mahyoub, Mohammed and Matrawy, Ashraf},
  booktitle={ International Conference on the Design of Reliable Communication Networks (DRCN)}, 
  title={On the Impact of Flooding Attacks on {5G} Slicing with Different {VNF} Sharing Configurations}, 
  year={2024},
  volume={},
  number={},
  pages={136-142}}

@article{Takapoui2020,
author = {Reza Takapoui and Nicholas Moehle and Stephen Boyd and Alberto Bemporad},
title = {A simple effective heuristic for embedded mixed-integer quadratic programming},
journal = {International Journal of Control},
volume = {93},
number = {1},
pages = {2--12},
year = {2020},
publisher = {Taylor \& Francis}
}

@ARTICLE{Abdelsadek2023,
  author={Abdelsadek, Mohammed Y. and Chaudhry, Aizaz U. and Darwish, Tasneem and Erdogan, Eylem and Karabulut-Kurt, Gunes and Madoery, Pablo G. and Ben Yahia, Olfa and Yanikomeroglu, Halim},
  journal={IEEE Transactions on Communications}, 
  title={Future Space Networks: Toward the Next Giant Leap for Humankind}, 
  year={2023},
  volume={71},
  number={2},
  pages={949-1007}}

@ARTICLE{Pachler2024,
  author={Pachler, Nils and Crawley, Edward F. and Cameron, Bruce G.},
  journal={IEEE Wireless Communications Letters}, 
  title={Flooding the Market: Comparing the Performance of Nine Broadband Megaconstellations}, 
  year={2024},
  volume={13},
  number={9},
  pages={2397-2401}}

@ARTICLE{Qiufen2024,
  author={Xia, Qiufen and Wang, Guijie and Xu, Zichuan and Liang, Weifa and Xu, Zhou},
  journal={IEEE Transactions on Mobile Computing}, 
  title={Efficient Algorithms for Service Chaining in {NFV}-Enabled Satellite Edge Networks}, 
  year={2024},
  volume={23},
  number={5},
  pages={5677-5694}}

@ARTICLE{Laniewski2025,
  author={Laniewski, Dominic and Lanfer, Eric and Aschenbruck, Nils},
  journal={IEEE Open Journal of the Communications Society}, 
  title={Measuring Mobile Starlink Performance: A Comprehensive Look}, 
  year={2025},
  volume={6},
  number={},
  pages={1266-1283}}

@ARTICLE{Pachler2026,
  author={Pachler, Nils and Crawley, Edward F. and Cameron, Bruce G.},
  journal={IEEE Transactions on Vehicular Technology}, 
  title={Mixed Integer Linear Programming for Ground Infrastructure and Inter-Satellite Link Design in Satellite Constellations}, 
  year={2026},
  volume={75},
  number={1},
  pages={796-810}}

@ARTICLE{Wang2022,
  author={Wang, Peng and others},
  journal={IEEE Transactions Wireless Communication}, 
  title={Enhancing {E}arth Observation Throughput Using Inter-Satellite Communication}, 
  year={2022},
  volume={21},
  number={10},
  pages={7990-8006}}

@ARTICLE{Wang2025,
  author={Wang, Ran and Cai, Lundan and Wu, Qiang and Niyato, Dusit},
  journal={IEEE Transactions on Mobile Computing}, 
  title={Service Function Chain Deployment With Intrinsic Dynamic Defense Capability}, 
  year={2025},
  volume={24},
  number={6},
  pages={5418-5432}}

@techreport{3gpp38811,
  title        = {Study on new radio ({NR}) to support non-terrestrial network},
  author  = {3rd Generation Partnership Project ({3GPP})},
  address      = {Sophia Antipolis, France},
  type         = {Technical Report},
  number       = {TR 38.811},
  version      = {15.4.0},
  year         = {2020}
}
 
\end{document}